\documentclass[a4paper]{article}
 
\usepackage[english]{babel}
\usepackage[utf8x]{inputenc}
\usepackage[T1]{fontenc}
\usepackage{graphicx,fancyhdr,amsmath,amssymb,amsthm,subfig,url,hyperref}

\usepackage[a4paper,top=3cm,bottom=2cm,left=3cm,right=3cm,marginparwidth=1.75cm]{geometry}

\usepackage{authblk}

\newtheorem{theorem}{Theorem}[section]
\newtheorem{lemma}[theorem]{Lemma}

\newtheorem{definition}{Definition}[section]

\usepackage{amsmath}
\usepackage{graphicx}
\usepackage[colorinlistoftodos]{todonotes}

\DeclareMathOperator*{\argmax}{argmax}

\allowdisplaybreaks

\title{Multi-unit Bilateral Trade \thanks{Supported by ERC Advanced
    Grant 321171 (ALGAME) and the EPSRC.}}

\author[1]{Matthias Gerstgrasser}
\author[1]{Paul W. Goldberg}
\author[2]{Bart de Keijzer}
\author[1]{Philip~Lazos}
\author[3]{Alexander Skopalik}

\affil[1]{Department of Computer Science, University of Oxford \protect\\
  \texttt{\{matthias.gerstgrasser, paul.goldberg, filippos.lazos\}@cs.ox.ac.uk}}
\affil[2]{School of Computer Science and Electronic
  Engineering, University of Essex \protect\\
  \texttt{b.dekeijzer@essex.ac.uk}}
\affil[3]{Department of Applied Mathematics, University of Twente
  \protect\\ \texttt{a.skopalik@utwente.nl}}

\begin{document}
\maketitle

\begin{abstract}
  We characterise the set of dominant strategy incentive compatible
  (DSIC), strongly budget balanced (SBB), and ex-post individually
  rational (IR) mechanisms for the multi-unit bilateral trade
  setting. In such a setting there is a single buyer and a single
  seller who holds a finite number $k$ of identical items. The
  mechanism has to decide how many units of the item are transferred
  from the seller to the buyer and how much money is transferred from
  the buyer to the seller. We consider two classes of valuation
  functions for the buyer and seller: Valuations that are increasing
  in the number of units in possession, and the more specific class of
  valuations that are increasing and submodular.

  Furthermore, we present some approximation results about the
  performance of certain such mechanisms, in terms of social welfare:
  For increasing submodular valuation functions, we show the existence
  of a deterministic $2$-approximation mechanism and a randomised
  $e/(1-e)$ approximation mechanism, matching the best known bounds
  for the single-item setting.
\end{abstract}

\section{Introduction}
Auctions form one of the most studied
applications of game theory and mechanism design. In an auction
setting, a single seller or auctioneer runs a pre-determined procedure
or mechanism (i.e., the auction) to sell one or more goods to the
buyers, and the buyers then have to strategise on the way they
interact with the auction mechanism. An auction setting is rather
restrictive in that it involves a single seller that is monopolistic
and is assumed to be non-strategic. While this is a sufficient
assumption in some cases, there are many applications that are more
complex: It is often realistic to assume that a seller expresses a
valuation for the items in her possession and that a seller wants to
maximise her profit. Such settings in which both buyers and sellers
are considered as strategic agents are known as \emph{two-sided
markets}, whereas auction settings are often referred to as
\emph{one-sided markets}.

The present paper falls within the area of mechanism design for
two-sided markets, where the focus is on designing satisfactory market
platforms or intermediation mechanisms that enable trade between
buyers and sellers. In general, the term ``satisfactory'' can be
tailored to the specific market under consideration, but nonetheless,
in economic theory various universal properties have been identified
and agreed on as important. The following three are the most
fundamental ones:
\begin{itemize}
\item \emph{Incentive Compatibility ((DS)IC)}: It must be a dominant
strategy for the agents (buyers and sellers) to behave truthfully,
hence not ``lie'' about their valuations for the items in the
market. This enables the market mechanism to make an informed decision
about the trades to be made.
\item \emph{Individual Rationality (IR)}: It must not harm the utility
of an agent to participate in the mechanism.
\item \emph{Strong Budget Balance (SBB)}: All monetary transfers that
the mechanism executes are among participating agents only. That is,
no money is injected into the market, and no money is burnt or
transferred to any agent outside of the market.
\end{itemize} This paper studies the capabilities of mechanisms that
satisfy these three fundamental properties above for a very simple
special case of a two-sided market.  {\em Bilateral trade} is the most
basic such setting comprising a buyer and a seller, together with a
single item that may be sold, i.e., transferred from the seller to the
buyer against a certain payment from the buyer to the seller. The
bilateral trade setting is a classical one: It was studied in the
seminal paper \cite{myerson1983efficient} and has been studied in
detail in various other publications in the economics
literature. Recent work in the Algorithmic Game Theory literature
\cite{bilateral,bicbilateral,CGKLRT} has focused on the welfare
properties of bilateral trade mechanisms. These works assume the
existence of {\em prior distributions} over the valuations of the
buyer and seller, that may be thought of as modelling an
intermediary's beliefs about the buyer's and seller's values for the
item.

The present paper studies a generalisation of the classical bilateral
trade setting by allowing the seller to hold multiple units
initially. These units are assumed to be of a single resource, so that
both agents only express valuations in terms of how many units they
have in possession. The final utility of an agent (buyer or seller) is
then determined by her valuation and the payment she paid or
received. We focus our study on characterising which mechanisms
satisfy the the above three properties and which of these feasible
mechanisms achieve a good social welfare (i.e., total utility of buyer
and seller combined).

Due to its simplicity, our setting is fundamental to any strategic
setting where items are to be redistributed or reallocated. Our
characterisation efforts show that all feasible mechanisms must belong
to a very restricted class, already for this very simple setting with
one buyer, one seller, and a relatively simple valuation
structure. The specific mechanisms we develop are very simple, and
suitable for implementation with very little communication complexity.

\paragraph{Our Contribution.}  Our first main contribution is a full
characterisation of the class of truthful, individually rational and
strongly budget balanced mechanisms in this setting. We do this
separately for two classes of valuation functions: submodular
valuations and general non-decreasing valuations. Section
\ref{sec:char} presents a high-level argument for the submodular
case. A full and rigorous formal proof for both settings is given in
\autoref{sec:proof-char-gener}. Essentially, for the general case, any
mechanism that aims to be truthful, strongly budget balanced and
individually rational can only allow the agents to trade a single
quantity of items at a predetermined price. The trade then only occurs
if both the seller and buyer agree to it. This leads to a very clean
characterization and has the added benefit of giving a robust, simple
to understand mechanism: the agents do not have to disclose their
entire valuation to the mechanism, and only have to communicate
whether they agree to trade one specific quantity at one specific
price.  For the submodular case, suitable mechanisms can be
characterised as specifying a per-unit price, and repeatedly letting
the buyer and seller trade an item at that price until one of them
declines to continue.

Secondly, we give approximation mechanisms for the social welfare
objective in the Bayesian setting in Section~\ref{sec:approx}, for the
case of submodular valuations. Theorem~\ref{thm:2approx} presents a
2-approximate deterministic mechanism. For randomised mechanisms, we
show a $e/(e-1)$-approximation in Theorem \ref{thm:eapprox}.

\paragraph{Related Literature.} The first approximation result for
bilateral trade was presented in \cite{mcafee}, where for the
single-item case the author proves that the optimal \emph{gain from
trade} can be 2-approximated by the \emph{median mechanism}, which is
a mechanism that sets the seller's median valuation as a fixed price
for the item, and trade occurs if and only if $p$ lies in between the
buyer's and seller's valuation and the buyer's valuation exceeds
$p$. The analysis in \cite{mcafee} is done under the assumption that
the seller's median valuation does not exceed the median valuation of
the buyer. The gain from trade is defined as the increase in social
welfare as a result of trading the item. \cite{bilateral} extended the
analysis of this mechanism by showing that it also $2$-approximates
the social welfare without the latter assumption on the medians.

In \cite{bilateral}, the authors furthermore consider the classical
bilateral trade setting (with a single item) and present various
mechanisms for it that approximate the optimal social welfare. Their
best mechanism achieves an approximation factor of $e/(e - 1)$. As in
the present paper, there are prior distributions on the traders'
valuations, and the quantity being approximated is the expectation
over the priors, of the optimal allocation of the item.

The weaker notion of Bayesian incentive compatibility is considered in
\cite{bicbilateral}, where the authors propose a mechanism in which
the seller offers a take-it-or-leave-it price to the buyer. They prove
that this mechanism approximates the harder \emph{gain from trade}
objective within a factor of $1/e$ under a technical albeit often
reasonable \emph{MHR condition} on the buyer's distribution.

The class of DSIC, IR, and SBB mechanisms for bilateral trade was
characterised in \cite{doubleauctions} to be the class of \emph{fixed
price mechanisms}. In the present work, we characterise this set of
mechanisms for the more general multi-unit bilateral trade setting,
thereby extending their result. The gain from trade arising from such
mechanisms was analysed in \cite{CRGKLT}.

Various recent papers analyse more general two-sided markets, where
there are multiple buyers and sellers, who hold possibly complex
valuations over the items in the market. \cite{CGKLRT} analyse a more
general scenario with multiple buyers, sellers, and multiple distinct
items, and use the same feasibility requirements as ours (DSIC, IR,
and SBB).  \cite{SHA18} have considered a similar setting but focus on
\emph{gains from trade (GFT)} (i.e., the increase in social welfare
resulting from reallocation of the items) instead of welfare. They
initially considered a multi-unit setting like ours (albeit with
multiple buyers and sellers), and they extend their work in
\cite{SHA18b} to allow multiple types of goods. They present a
mechanism that approximates the optimal GFT asymptotically in large
markets. \cite{dynamic} designs two-sided market mechanisms for one
seller and multiple buyers with a temporal component, where valuations
are correlated between buyers but independent across time steps. A
good approximation (of factor $1/2$) of the social welfare using the
more permissive notion of \emph{Bayesian} Incentive Compatibility
(BIC) was achieved by \cite{brustle}. Their optimality benchmark is
different from the one we consider as they compare their mechanism to
the best possible BIC, IR, and SBB mechanism. A very recent work,
\cite{babaioff}, proposes mechanisms that achieve social welfare
guarantees for both optimality benchmarks. \cite{feldman1} considers
optimizing the gains from trade in a two-sided market setting tailored
to online advertising platforms, and the authors extend this idea
further in \cite{feldman2} by considering two-sided markets in an
online setting.

The literature discussed so far aims to maximise welfare under some
budget-balance constraints. An alternative natural goal is to maximise
the intermediary's profit. This has been studied extensively starting
with a paper by Myerson and Satterthwaite \cite{myerson1983efficient},
which gives an analogue of Myerson's seminal result on optimal
auctions, for the independent priors case. Approximately optimal
mechanisms for that settings have further been
studied. \cite{deng2014revenue,niazadeh2014simple} The
correlated-priors case has been investigated from a computational
complexity perspective by \cite{gerstgrasser2016revenue}, as well as
links back to auction theory \cite{gerstgrasser2018complexity}. Two
adversarial, online variants of market intermediation were studied in
\cite{giannakopoulos2017online,koutsoupias2018online}.

\section{Preliminaries}\label{sec:prelims}

In a \emph{multi-unit bilateral trade} instance there is a buyer and
seller, where the seller holds a number of units of an item. This
number will be denoted by $k$. The buyer and seller each have a
\emph{valuation function} representing how much they value having any
number of units in possession. These valuation functions are denoted
by $v$ and $w$, respectively. Precisely stated, a valuation function
is a function $v : [k] \cup \{0\} \rightarrow \mathbb{R}_{\geq 0}$
where $v(0) = 0$. Note that we use the standard notation $[a]$, for a
natural number $a$, to denote the set $\{1, \ldots, a\}$. We denote by
$v$ the valuation function of the buyer, drawn from $f$, and we denote
by $w$ the valuation function of the seller, drawn from $g$. For $q
\in [k]$, the valuation $v(q)$ or $w(q)$ of an agent (i.e., buyer or
seller) expresses in the form of a number the extent to which he would
like to have $q$ units in his possession.

A mechanism $\mathbb{M}$ interacts with the buyer and the seller and
decides, based on this interaction, on an \emph{outcome}. An outcome
is defined as a quadruple $(q_B,q_S,\rho_B,\rho_S)$, where $q_B$ and
$q_S$ denote the numbers of items allocated to the buyer and the
seller respectively, such that $q_B+q_S = k$. Moreover, $\rho_B$ and
$\rho_S$ denote the payments that the mechanism charges to the buyer
and seller respectively. Note that typically the payment of the seller
is negative since he will get money in return for losing some items,
while the payment of the buyer is positive since he will pay money in
return for obtaining some items. Let $\mathcal{O}$ be the set of all
outcomes. For brevity we will often refer to an outcome simply by the
number of units traded $q_B$.

Formally, a mechanism is a function $\mathbb{M} :\Sigma_B \times
\Sigma_S \rightarrow \mathcal{O}$, where $\Sigma_B$ and $\Sigma_S$
denote strategy sets for the buyer and seller. A \emph{direct
revelation} mechanism is a mechanism for which $\Sigma_B$ and
$\Sigma_S$ consists of the class of valuation functions that we want
to consider. That is, in such mechanisms, the buyer and seller
directly report their valuation function to the mechanism, and the
mechanism decides an outcome based on these reports. We want to define
our mechanism in such a way that there is a dominant strategy for the
buyer and seller, under the assumption that their valuation functions
are in a given class $\mathcal{V}$. It is well known (see
e.g. \cite{boergers}) that then we may restrict our attention to
direct revelation mechanisms in which the dominant strategy for the
buyer and seller is to report the valuation functions that they
hold. Such mechanisms are called \emph{dominant strategy incentive
compatible (DSIC)} for $\mathcal{V}$. We will assume from now on that
$\mathbb{M}$ is a direct revelation mechanism. In this paper, we
consider for $\mathcal{V}$ two natural classes of valuation functions:
\begin{itemize}
\item Monotonically increasing submodular functions, i.e., valuation
functions $v$ such that for all $x,y \in [k]$ where $x < y$ it holds
that $v(x) - v(x-1) \geq v(y) - v(y-1)$ and $v(x) < v(y)$. This
reflects a common phenomenon observed in many economic settings
involving identical goods: Possessing more of a good is never
undesirable, but the increase in valuation still goes down as the held
amount increases. For a monotonically increasing submodular function
$v$ and number of units $x \in [k]$, we denote by $\tilde{v}(x)$ the
\emph{marginal} valuation $v(x)-v(x-1)$. Thus, it holds that
$\tilde{v}(x) \geq \tilde{v}(y)$ when $x < y$.  
\item Monotonically increasing functions, i.e., valuation functions
$v$ such that $v(x) < v(y)$ for all $x < y$, where $x,y \in [k]$. 
\end{itemize}

Besides the DSIC requirement, there are various additional properties
that we would like our mechanism to satisfy.
\begin{itemize}
\item Ideally, our mechanism should be \emph{strongly budget balanced
(SBB)}, which means that for any outcome $(q_B,q_S,\rho_B,\rho_S)$
that the mechanism may output it holds that $\rho_B = -\rho_S$. This
requirement essentially states that all money transferred is between
the buyer and the seller only.
\item Additionally, we want that running the mechanism never harms the
buyer and the seller. This requirement is known as \emph{(ex-post)
individual rationality (IR)}. Note that when $v$ and $w$ are the
valuation functions of the buyer and the seller, then the initial
utility of the buyer is $0$ and the initial utility of the seller is
$w(k)$. Thus, a mechanism $\mathbb{M}$ is individually rational if for
the outcome $\mathbb{M}(v,w) = (q_B,q_S,\rho_B,\rho_S)$ it always
holds that $v(q_B) - \rho_B \geq 0$ and $w(q_S) - \rho_S \geq w(k)$.
\item We would like the mechanism to return an outcome for which the
total utility is high. That is, we want the mechanism to maximise the
sum of the buyer's and seller's utility, which is equivalent to
maximizing the sum of valuations $v + w$ when strong budget balance
holds.
\end{itemize}

We characterise in Section \ref{sec:char} the class of DSIC, SBB, IR
mechanisms for both valuation classes. In Section \ref{sec:approx}, we
subsequently provide various approximation results on the quality of
the solution output by some of these mechanisms. For these results, we
assume the standard \emph{Bayesian setting}: The mechanism has no
knowledge of the buyer's and seller's precise valuation, but knows
that these valuations are drawn from known probability distributions
over valuation functions. Our approximation results provide mechanisms
that guarantee a certain outcome quality (which is measured in terms
of \emph{social welfare}, defined in Section \ref{sec:approx}) for
arbitrary distributions on the valuation functions.

Formally, in the Bayesian setting, a multi-unit bilateral trade
instance is a pair $(f,g,k)$, where $k \in \mathbb{N}$ is the total
number of units that the seller initially has in his possession, and
$f$ and $g$ are probability distributions over valuation functions of
the buyer and the seller respectively. Note that we do not impose any
further assumptions on these probability distributions.

\section{Characterisation}\label{sec:char}
In \cite{doubleauctions} the authors prove that every DSIC, IR, SBB
mechanism for classical bilateral trade (i.e. the case where $k=1$) is
a \emph{fixed price mechanism}: That is, the mechanism is parametrised
by a price $p \in \mathbb{R}_{\geq 0}$ such that the buyer and seller
trade if and only if the buyer's valuation exceeds the price and the
price exceeds the seller's valuation. Moreover, in case trade happens,
the buyer pays $p$ to the seller. In this paper we characterise the
set of DSIC, IR, and WBB mechanisms for multi-unit bilateral trade,
and we thereby generalise the characterisation of
\cite{doubleauctions}.

\begin{theorem}\label{thm:char_simple} Any mechanism that satisfies DSIC, IR
and SBB must be a \emph{sequential posted price mechanism with a fixed
per-unit price $p$, potentially with bundling}, which we will refer to
as a \emph{multi-unit fixed price mechanism}. Such a mechanism
iteratively proposes a quantity $q$ of units to both the buyer and
seller simultaneously, which the seller and buyer can choose to either
accept or reject. If both agents accept, $q$ additional units are
reallocated from the seller to the buyer, the buyer pays $pq$ to the
seller, and the mechanism may then either proceed to the next
iteration or terminate. If one of the two agents rejects, the
mechanism terminates. Quantity $q$ may vary among iterations, but must
be pre-determined prior to execution of the mechanism.

For increasing submodular valuations, any number of iterations is
allowed. For general increasing valuations, the mechanism is further
restricted to execute only one iterations (or equivalently, it may
only offer one bundle for a fixed price).
\end{theorem}

In simple terms, our result states that for the submodular valuations
case, the only thing to be done truthfully in this setting is to set a
fixed per-unit price $p$, and ask the buyer and seller if they want to
trade one or several units of the good at per-unit price $p$. This
repeats until one agent rejects. In the general monotone case this is
further restricted to a single such proposed trade. The following is a
brief high-level (informally stated) argument of the proof of Theorem
\ref{thm:char_simple} for the submodular setting. We refer the reader to
\autoref{sec:proof-char-gener} for the complete proof.

\begin{lemma} \label{prop:fixed-price-DSIC} All prices must be fixed
in advance, and cannot depend on the bid / valuation of neither the
seller nor the buyer.
\end{lemma}
\begin{proof} This follows immediately from DSIC and SBB: By DSIC, for
any outcome, the price charged to the buyer can't depend on the
buyer's bid, otherwise one can construct scenarios in which the price
charged by the buyer could be manipulated to the buyer's benefit by
misreporting the bid. The same holds for the seller. By SBB the
payment of the buyer completely determines the payment of the seller
(the payment is simply negated) so neither payment can depend on
either's bid.
\end{proof}

\begin{theorem} \label{prop:const-price-DSIC} Suppose in a DSIC, SBB,
IR mechanism the price for the outcome in which $q$ units are traded
is $q p$ for a fixed per-unit price for all potential outcomes. Then
the allocation chosen for a given pair of valuation functions is the
one arising when asking bidders sequentially if they want to trade one
unit (or a bundle of units), until one rejects.
\end{theorem}
\begin{proof} 
To see this, consider the seller's utility function
$u_{s}(q) = q\cdot p + w(k - q)$ and the buyer's utility function
$u_{b} (q) = v(q) - q\cdot p$, if $q$ units would be traded at unit
price $p$. Since both valuation functions are concave, it is easy to
see that both utility functions are concave, and each has a single
peak (one or more equal adjacent maxima, and no further local
maxima). Furthermore they both start at 0, and once either of them
becomes negative, it stays negative. Suppose we sequentially ask both
bidders if they want to trade one unit for price $p$, until one
rejects. Then the quantity traded is
$\min(\argmax(u_{s}), \argmax{(u_{b})})$, i.e. the first of the two
peaks. If the mechanism iteratively proposes them bundles
$q_1, q_2, \ldots$, then the same expression on the traded quantity
would apply, but with the utility functions restricted to the domain
$\left\{0, q_1, q_1+q_2, \ldots \right\}$. If we ask them about the
big all-$k$-item bundle, we would choose the bundle outcome iff
$u(k) > u(0)$, for both, and $0$ if for either of them $u(0) > u(k)$,
i.e. if one (the first) of the peaks of the two utility functions
restricted to $\left \{ 0,k \right\}$ is at $0$.
	
Now, DSIC means that for any bid of the opposing agent, the agent
cannot get anything better than what she gets by telling the truth. If
the quantity traded by the mechanism would be larger than
$\min(\argmax(u_s), \argmax{(u_b)})$, then the bidder with the lowest
peak could improve her utility by claiming that all outcomes higher
than her peak are wholly unacceptable (utility less than 0) to them;
by IR, the mechanism would then be forced to trade the quantity at the
first peak. If, on the other hand, the traded quantity would be less
than the quantity of the first peak, then both players would gain by
lying, in order to make the mechanism choose to trade a higher
quantity (if such a quantity is at all present in the mechanism's set
of tradeable quantities.)
\end{proof}

\begin{theorem} \label{prop:var-price-DSIC} In a DSIC, SBB, IR
mechanism, all potential outcomes, i.e., (quantity,price)-pairs, must
have the same per-unit price.
\end{theorem}
\begin{proof} Suppose two outcomes have different per-unit
prices. W.l.o.g. suppose for $q_1 < q_2$, $p_1/q_1 < p_2/q_2$,
i.e. the per-unit price is higher in the larger allocation. Then there
exists a valuation function $v_{s1}$ for the seller in which the
seller prefers outcome $q_2$ over $q_1$, but both give positive
utility; and there exists another valuation function $v_{s2}$ that
gives negative utility for $q_1$, but the same utility for
$q_2$. I.e. $ 0 < u_{s1} (q_1) < u_{s2} (q_2)$ but $u_{s2} (q_1) < 0 <
u_{s2} (q_2) = u_{s1} (q_2)$. Now if for a given buyer's valuation,
the chosen outcome given $v_{s1}$ is $q_1$, then the seller would have
an incentive to misreport $v_{s2}$, making outcome $q_1$ unavailable
to the mechanism due to IR, thus making it choose $q_2$. Vice versa,
if per-unit prices are decreasing, the same argument works for the
buyer.
\end{proof}

Together, these three results give a full characterisation of the
class of DSIC, IR, SBB mechanisms in this setting, although in our
full formal proof that we provide in \autoref{sec:proof-char-gener}, we
need to take into account many further technical obstacles and
details. There is, in particular, a \emph{tie-breaking rule} present,
that takes into account what should happen when the buyer or seller
would be indifferent among multiple possible quantities, or when they
would get a utility of 0 given the proposed prices and quantities.

For the case of general monotone valuations, any such mechanism must
be further restricted to offering only a single outcome (other than
no-trades) to the bidders. The complete proof can be found in
\autoref{sec:proof-char-gener}.

\section{Approximation Mechanisms}\label{sec:approx}
In this section we study the design of DSIC, IR, SBB mechanisms that
optimise the social welfare, i.e., the sum of the buyer's and seller's
valuation. From Theorem~\ref{thm:char_simple}, our characterization states
that such a mechanism needs to be a multi-unit fixed price mechanism,
so that the design challenge lies in an appropriate choice of
unit-price $p$ and quantities offered at each iteration of the
mechanism.

We focus on the case of increasing submodular valuations. Obviously,
every item traded can only increase the social welfare. Therefore,
given that the objective is to maximise it, we repeatedly offer a
single item for trade.\footnote{Also, with respect to our tie-breaking
  rule mentioned at the end of the last section: We simply employ the
  tie breaking rule that favours the highest quantity to trade, which
  is the dominant choice when it comes to maximising social welfare.}
The challenge lies thus in determining the right unit price $p$. It is
easy to see that no sensible analysis can be done if absolutely
nothing is known about the valuation functions of the buyer and
seller. Therefore, we assume a \emph{Bayesian setting}, as introduced
in Section \ref{sec:prelims} in order to model that the mechanism
designer has statistical knowledge about the valuations of the two
agents: The buyer's (and seller' valuation is assumed to be unknown to
the mechanism, but is assumed to be drawn from a probability
distribution $f$ (and $g$) which is public knowledge. We show that we
can now determine a unit price that leads to a good social welfare in
expectation.

For a valuation function $v$ of the buyer, we write $\hat{v}$ to
denote the \emph{marginal increase function} of $v$:
$\hat{v}(q) = v(q) - v(q-1)$ for $q \in [k]$. Thus, $\hat{v}$ is a
non-increasing function. Similarly, for a valuation function $w$ of
the seller, we write $\check{w}$ to denote the \emph{marginal decrease
  function} of $w$: $\check{w}(q) = w(k-q+1) - w(k-q)$, for
$q \in [k]$, so that $\check{w}$ is a non-decreasing function. Thus,
for all $q \in [k]$, the increase in social welfare as a result of
trading $q$ items as opposed to $q-1$ items is
$\hat{v}(q) - \check{w}(q)$. Note that therefore if $v$ and $w$ are
increasing submodular valuation functions of the buyer and seller
respectively, then the social welfare is maximised by trading the
maximum number of units $q$ such that $\hat{v}(q) > \check{w}(q)$.  We
measure the quality of a mechanism on a bilateral trade instance
$(f,g,k)$ as the factor by which its expected social welfare is
removed from the expected optimal social welfare $OPT(f,g,k)$ that
would be attained if the buyer and seller would always trade the
maximum profitable amount: 
\begin{align*}
  OPT(f,g,k)
  &=\mathop{\mathbf{E}}_{v \sim f, w \sim g}\left[w(k) + \sum_{q =
  1}^{\max\{q' : \hat{v}(q') >  
  \check{w}(q')\}} (\hat{v}(q) - \check{w}(q))\right] \\
  &=\mathop{\mathbf{E}}_{v \sim f, w \sim g}\left[\sum_{q = 1}^k \check{w}(q) + \sum_{q = 
    1}^{\max\{q' : \hat{v}(q') \geq \check{w}(q')\}} (\hat{v}(q) - \check{w}(q))\right] 
\end{align*}
 For $q \in [k]$ and a seller's valuation function $w$, we
denote by $GFT(v,w,q)$ the value
$\max\{0, \hat{v}(q) - \check{w}(q)\}$ (where ``GFT'' is intended to
stand for ``Gain From Trade''). Note that $GFT(v,w,q)$ is
non-increasing in $q$ and that $OPT(f,g,k)$ can be written as 
\begin{equation*}
OPT(f,g,k) = \sum_{q = 1}^k \mathbf{E}_{w \sim g}[\check{w}(q) + GFT(v,w,q)]. 
\end{equation*}
 Note that a social welfare as high as opt $OPT(f,g,k)$ can
typically not be attained by any DSIC, IR, SBB mechanism. However, it
is still a natural benchmark for measuring the performance of such a
mechanism, and we will see next that there exists such a mechanism
that achieves a social welfare that is guaranteed to approximate
$OPT(f,g,k)$ to within a constant factor. In particular, for a
mechanism $\mathbb{M}$, let $q_{\mathbb{M}}(v,w)$ be the number of
items that $\mathbb{M}$ trades on reported valuation profiles $(v,w)$,
and define 
\begin{align*}
SW(\mathbb{M},(g,f,k)) = \mathbf{E}_{v \sim f, w \sim g}[v(q_{\mathbb{M}}(v,w)) + 
v(k-q_{\mathbb{M}(v,w)})]
\end{align*}
 as the expected social welfare of mechanism
$\mathbb{M}$. We say that $\mathbb{M}$ achieves an
\emph{$\alpha$-approximation to the optimal social welfare}, for
$\alpha > 1$, iff $OPT(g,f,k)/SW(\mathbb{M},(g,f,k)) < \alpha$.

We show next that the multi-unit fixed price mechanism where $p$ is
set such that
$$\sum_{q=1}^k \mathbf{Pr}_{w \sim g}[\check{w}(q) \leq p] = k/2$$
achieves a $2$-approximation to the optimal social welfare.

\begin{theorem}\label{thm:2approx}
  Let $(f,g,k)$ be a multi-unit bilateral trade instance where the
  supports of $f$ and $g$ contain only increasing submodular
  functions. Let $\mathbb{M}$ be the multi-unit bilateral trade
  mechanism where at each step one item is offered for trade at price
  $p = \sum_{q=1}^k \mathbf{Pr}_{w \sim g}[\check{w}(q) \leq p] =
  k/2$, until either agent reject the offer (informally: $p$ is the
  price such that the seller is expected to accept to trade half of
  his units at price $p$). Mechanism $\mathbb{M}$ achieves a
  $2$-approximation to the optimal social welfare.
\end{theorem}
\begin{proof}
  Let $v$ be an arbitrary buyer's valuation function. We show that the
  mechanism achieves a $2$-approximation if $f$ is the distribution
  having only $v$ in its support, and hence $v$ is the buyer's
  valuation with probability $1$. It suffices to prove the claim under
  this assumption, because the unit-price $p$ depends on distribution
  $g$ only. Hence, if $\mathbb{M}$ achieves the claimed social welfare
  guarantee for every fixed buyer's valuation function, then it also
  achieves this guarantee for every distribution on the buyer's
  valuation. For ease of notation, we will abbreviate
  $SW(\mathbb{M},(f,g,k))$ to simply $SW$ and we let
  $\ell = \max\{q : \hat{v}_k(q) \geq p\}$ be the highest quantity
  that the buyer would like to trade at unit-price $p$.  In the
  remainder of the proof, we will omit the subscript $w \sim g$ from
  the expected value operator.

  We first observe that $SW$ can be written as follows, where we write
  $\mathbf{1}[\cdot]$ to denote the indicator function and $E_q$ for
  the event that $\hat{v}(q) \geq p \geq \check{w}(q)$.
  \begin{align}
    SW &=
         \mathbf{E}\left[\sum_{q = 1}^k (\check{w}(q) +
    \mathbf{1}[E_q]GFT(v,w,q))\right] \notag\\  
    &= \mathbf{E}\left[\sum_{q = 1}^\ell (\check{w}(q) +
      \mathbf{1}[E_q]GFT(v,w,q))\right] + \mathbf{E}\left[\sum_{q =
      \ell+1}^k \check{w}(q)\right] \label{eq:swexp}  
  \end{align}
   We will bound these last two expected values separately
  in terms of $OPT(f,g,k)$, and subsequently we will combine the two
  bounds to obtain the desired approximation factor.

  We start with the quantities up to $\ell$, for which first rewrite
  the expression as follows.
\begin{align*}
  \mathbf{E}\left[\sum_{q = 1}^\ell (\check{w}(q) + \mathbf{1}[E_q]GFT(v,w,q))\right]
  &= \sum_{q = 1}^\ell \mathbf{E}[ \check{w}(q)] +  \sum_{q = 1}^\ell 
    \mathbf{Pr}[E_q]\mathbf{E}[GFT(v,w,q))\ |\ E_q] \\
  &   = \sum_{q = 1}^\ell \mathbf{E}[ \check{w}(q)] +  \sum_{q =
    1}^\ell \mathbf{Pr}[E_q] \mathbf{E}[GFT(v,w,q))\ |\ E_q]. 
\end{align*}
 Now, observe that
$\mathbf{Pr}[E_q] = \mathbf{Pr}[p \geq \check{w}(q)]$ for quantities
$q \leq \ell$. Since
$\sum_{q = 1}^k \mathbf{Pr}[p \geq \check{w}(q)] = k/2$ and
$\mathbf{Pr}[p \geq \check{w}(q)]$ is decreasing in $q$, this implies
that
$\sum_{q = 1}^\ell \mathbf{Pr}[E_q] = \sum_{q = 1}^\ell \mathbf{Pr}[p
\geq \check{w}(q)] \geq \ell/2$. Using additionally the fact that
$\mathbf{E}[GFT(v,w,q))\ |\ E_q]$ is also non-increasing in $q$, we
obtain the following bound.
\begin{align}
  \mathbf{E}\left[\sum_{q = 1}^\ell (\check{w}(q) + 
  \mathbf{1}[E_q]GFT(v,w,q))\right]
  &\geq \sum_{q = 1}^\ell \mathbf{E}[ \check{w}(q)] + \frac{\sum_{q = 1}^\ell 
    \mathbf{Pr}[E_q]}{\ell}\sum_{q=1}^\ell \mathbf{E}[GFT(v,w,q))\ |\
    E_q] \notag\\
  &\geq \sum_{q = 1}^\ell \mathbf{E}[ \check{w}(q)] +  \frac12
    \sum_{q=1}^\ell \mathbf{E}[GFT(v,w,q))\ |\ E_q]\notag\\ 
  &\geq \sum_{q = 1}^\ell \mathbf{E}[ \check{w}(q)] +  \frac12
    \sum_{q=1}^\ell \mathbf{E}[GFT(v,w,q))]\notag\\ 
  &\geq \frac12\sum_{q = 1}^\ell \mathbf{E}[\check{w}(q) +
    GFT(v,w,q)]\label{eq:firstbound}
\end{align}

For the quantities higher than $\ell$, we first observe that
non-increasingness of $\mathbf{Pr}[\check{w}(q) < p]$ in the quantity
$q$ implies that $\mathbf{Pr}[\check{w}(q) > p]$ is non-decreasing in
$q$. Moreover, $\sum_{q = 1}^k \mathbf{Pr}[\check{w}(q) \leq p] = k/2$
means that
$\sum_{q = 1}^k \mathbf{Pr}[\check{w}(q) > p] = \sum_{q = 1}^k
\mathbf{Pr}[\check{w}(q) \leq p]$, hence it holds that
$\sum_{q = \ell+1}^k \mathbf{Pr}[\check{w}(q) > p] \geq \sum_{q = 1}^k
\mathbf{Pr}[\check{w}(q) \leq p]$. Therefore, we derive
\begin{align}
  \mathbf{E}\left[\sum_{q = \ell+1}^k \check{w}(q)\right]
  &= \frac12 \sum_{q = \ell+1}^k \mathbf{E}[\check{w}(q)] +
    \frac12 \sum_{q = \ell+1}^k \mathbf{E}[\check{w}(q)] \notag \\ 
  & \geq \frac12 \sum_{q = \ell+1}^k \mathbf{E}[\check{w}(q)] +
    \frac12 \sum_{q = \ell+1}^k \mathbf{E}[\check{w}(q)\ |\
    \check{w}(q) > p]\mathbf{Pr}[\check{w}(q) > p] \notag \\ 
  & \geq \frac12 \sum_{q = \ell+1}^k \mathbf{E}[\check{w}(q)] +
    \frac12 \sum_{q = \ell+1}^k \hat{v}(q)\mathbf{Pr}[\check{w}(q) >
    p] \notag \\ 
  & \geq \frac12 \sum_{q = \ell+1}^k \mathbf{E}[\check{w}(q)] +
    \frac12 \sum_{q = \ell+1}^k \mathbf{E}[GFT(v,w,q)] \notag \\ 
  & \geq \frac12 \sum_{q = \ell+1}^k \mathbf{E}[\check{w}(q) +
    GFT(v,w,q)] \label{eq:secondbound} ,  
\end{align}
 where the second inequality holds because $\check{w}(q)$
conditioned on $\check{w}(q) > p$ is always higher than $\hat{v}(q)$
which does not exceed $p$. Moreover, the third inequality follows
because
$\mathbf{E}[GFT(v,w,q)] = \mathbf{E}[(\hat{v}(q) -
\check{w}(q))\mathbf{1}(\hat{v}(q) > \check{w}(q))] \leq
\mathbf{E}[\hat{v}(q)\mathbf{1}(\hat{v}(q) > \check{w}(q))] \leq
\mathbf{E}[\hat{v}(q)\mathbf{1}(p > \check{w}(q))] =
\hat{v}(q)\mathbf{Pr}[p > \check{w}(q)]$.

We now use (\ref{eq:firstbound}) and (\ref{eq:secondbound}) to bound
(\ref{eq:swexp}) and obtain the desired inequality
\begin{equation*}
SW \geq \frac12 \sum_{q = 1}^k \mathbf{E}[\check{w}(q) + GFT(v,w,q)] =
\frac{OPT(f,g,k)}{2}, 
\end{equation*}
which proves the claim. 
\end{proof}
The above $2$-approximation mechanism is deterministic. We show next
that we can do better if we allow randomisation: Consider the
\emph{Generalized Random Quantile Mechanism}, or $\mathbb{M}_{G}$,
which draws a number $x$ in the interval $[1/e,1]$ where the CDF is
$\ln(ex)$ for $x \in [1/e,1]$. The mechanism then sets a unit price
$p(x)$ such that
$\mathbf{E}_{w}[\max\{q : w(q) \geq qp(x)\}] = \sum_{q=1}^k
\mathbf{Pr}_{w}[\check{w}(q) \leq p(x)] = xk$, repeatedly offering
single item trades as before. In words, the price is set such that the
expected number of units that the seller is willing to sell, is an $x$
fraction of the total supply, where $x$ is randomly drawn according to
the probability distribution just defined. This randomised mechanism
satisfies DSIC, IR, and SBB, because it is simply a distribution over
multi-unit fixed price mechanisms.  Note that this mechanism is also a
generalisation of a previously proposed mechanism: In
\cite{bilateral}, the authors define the special case of this
mechanism for a single item, and call it the
$\emph{Random Quantile Mechanism}$. They show that it achieves a
$e/(e - 1)$-approximation to the social welfare, and we will prove
next that this generalisation preserves the approximation factor,
although the proof we provide for it is substantially more complicated
and requires various additional technical insights.
\begin{theorem}\label{thm:eapprox}
  Let $(f,g,k)$ be a multi-unit bilateral trade instance where the
  supports of $f$ and $g$ contain only increasing submodular
  functions. The Generalised Random Quantile Mechanism $\mathbb{M}_G$
  achieves a $e/(e-1)$-approximation to the optimal social welfare.
\end{theorem}
\begin{proof}
  As in the proof of Theorem \ref{thm:2approx}, we fix a valuation
  function $v$ for the buyer.  It suffices to prove the claim under
  this assumption, because the unit-price $p$ depends on distribution
  $g$ only. For ease of notation, we will again abbreviate
  $SW(\mathbb{M}_G,(f,g,k))$ to simply $SW$.

  We first rewrite $OPT(f,g,k)$ as follows: 
\begin{align}
OPT(f,g,k) &= \sum_{q=1}^k
             \mathbf{E}_w[\max\{\hat{v}(q),\check{w}(q)\}]\notag\\ 
  &= \sum_{q=1}^k \mathbf{E}_w[\hat{v}(q)]
    + \sum_{q=1}^k \mathbf{E}_w[(\check{w}(q) -
    \hat{v}(q))\mathbf{1}[\check{w}(q) \geq \hat{v}(q)]] \notag \\ 
  & = \sum_{q=1}^k \hat{v}(q) +
    \sum_{q = 1}^k \mathbf{E}_w[\check{w}(q) - \hat{v}(q)\ |\
    \check{w}(q) \geq \hat{v}(q)]
  \cdot \mathbf{Pr}_w[\check{w}(q) \geq \hat{v}(q)] \notag \\
  &= \sum_{q=1}^k \hat{v}(q) + \sum_{q=1}^k
    (\mathbf{E}_w[\check{w}(q)\ |\ \check{w}(q) \geq \hat{v}(q)] -
    \hat{v}(q))\cdot \mathbf{Pr}_w[\check{w}(q) \geq
    \hat{v}(b)] \label{eq:optbound}. 
\end{align}

In the remainder of the proof, we will derive a lower bound of
$(1-1/e)$ times the expression $(\ref{eq:optbound})$ on $SW$, which
implies our claim. We first observe that $SW$ can be bounded and
rewritten as follows.  
\begin{align}
  SW &= \sum_{q=1}^k \mathbf{E}_w[\check{w}(q)\mathbf{1}[\check{w}(q) \geq 
       \hat{v}(q)]]
       + \sum_{q=1}^k \mathbf{Pr}_w[\check{w}(q) < 
       \hat{v}(q)]  \mathbf{E}_{w,x}[\hat{v}(q)\mathbf{1}[p(x) \in 
       [\check{w}(q),\hat{v}(q)]]\notag\\ 
     &\qquad+ \check{w}(q)\mathbf{1}[p(x) \not\in 
       [\check{w}(q),\hat{v}(q)]]\ |\ \check{w}(q) < \hat{v}(q) ] \notag \\
     &\geq \sum_{q=1}^k \mathbf{E}_w[\hat{v}(q)\mathbf{1}[\check{w}(q) \geq 
       \hat{v}(q)] + (\check{w}(q) - \hat{v}(q) )
       \mathbf{1}[\check{w}(q) \geq \hat{v}(q)]]  
       \notag \\
     &\qquad+ \sum_{q=1}^k \mathbf{E}_{w,x}[\hat{v}(q)\mathbf{1}[p(x) \in 
       [\check{w}(q),\hat{v}(q)]]\ |\ \check{w}(q) < \hat{v}(q) ]]\notag
       \cdot
       \mathbf{Pr}_w[\check{w}(q) < \hat{v}(q)] \notag \\
     &= \sum_{q=1}^k \hat{v}(q)\mathbf{Pr}_w[\check{w}(q) \geq
       \hat{v}(q)] \notag\\
     &\qquad+ \sum_{q=1}^k (\mathbf{E}_w[\check{w}(q)\ |\ \check{w}(q) \geq 
       \hat{v}(q)] - \hat{v}(q)) \mathbf{Pr}[\check{w}(q) \geq \hat{v}(q)] \notag \\
     &\qquad+ \sum_{q=1}^k \mathbf{E}_{w,x}[\hat{v}(q)\mathbf{1}[p(x) \in 
       [\check{w}(q),\hat{v}(q)]\ |\ \check{w}(q) < \hat{v}(q) ]]
       \mathbf{Pr}_w[\check{w}(q) <  
       \hat{v}(q)] \notag \\
     & = \sum_{q=1}^k \hat{v}(q)\mathbf{Pr}_w[\check{w}(q) \geq \hat{v}(q)]  \notag \\
     &\qquad + 
       \sum_{q=1}^k \hat{v}(q)\mathbf{Pr}_{w,x}[p(x) \in [\check{w}(q),\hat{v}(q)]\ |\ 
       \check{w}(q) < \hat{v}(q) ]] \cdot
       \mathbf{Pr}_w[\check{w}(q) < \hat{v}(q)] 
       \label{eq:firstpart} \\
     &\qquad+ \sum_{q=1}^k (\mathbf{E}_w[\check{w}(q)\ |\ \check{w}(q) \geq
                          \hat{v}(q)] -
                          \hat{v}(q))\mathbf{Pr}[\check{w}(q) \geq
                          \hat{v}(q)] . \notag 
\end{align}

Next, we bound the first part (\ref{eq:firstpart}) of the last expression, i.e., excluding the last summation. 

\begin{align}
  (5) &\leq \sum_{q=1}^k \hat{v}(q) \mathbf{Pr}_w[\check{w}(q) \geq
        \hat{v}(q)] + \sum_{q=1}^k \hat{v}(q)\mathbf{Pr}_w[\check{w}(q)
        < \hat{v}(q)] \cdot  \frac{\int_{1/e}^{z : p(z) = \hat{v}(q)}
        \mathbf{Pr}_w[\check{w}(q) \leq p(x)] \frac{1}{x}
        dx}{\mathbf{Pr}_w[\check{w}(q) < \hat{v}(q)]} \notag \\ 
      & = \sum_{q=1}^k \hat{v}(q) \mathbf{Pr}_w[\check{w}(q) \geq
        \hat{v}(q)] \notag + \int_{1/e}^{z : p(z) = \hat{v}(q)}
        \left(\sum_{q=1}^k \hat{v}(q)\mathbf{Pr}_w[\check{w}(q) \leq
        p(x)]\right) \frac{1}{x} dx \notag \\ 
      &  \geq \sum_{q=1}^k \hat{v}(q) \mathbf{Pr}_w[\check{w}(q) \geq
        \hat{v}(q)]  +\int_{1/e}^{z : p(z) = \hat{v}(q)} \sum_{q=1}^k
        \hat{v}(q)\frac{kx}{k}\frac{1}{x} dx \notag \\ 
      &  = \sum_{q=1}^k \hat{v}(q) \mathbf{Pr}_w[\check{w}(q) \geq
        \hat{v}(q)] + \sum_{q=1}^k \hat{v}(q) \int_{1/e}^{z : p(z) =
        \hat{v}(q)} 1 dx \notag \\ 
      &  = \sum_{q=1}^k \hat{v}(q) \mathbf{Pr}_w[\check{w}(q) \geq
        \hat{v}(q)] + \sum_{q=1}^k \hat{v}(q)
        (\mathbf{Pr}[\check{w}(q) < \hat{v}(q)] - \frac{1}{e}) \notag \\ 
      &  = (1 - 1/e)\sum_{q=1}^k \hat{v}(q) \label{eq:boundonfirstpart},
\end{align}
where for the inequality we used that both $\hat{v}(q)$ and
$\mathbf{Pr}_{w}[\check{w}(q) < \hat{v}(q)]$ are non-increasing in
$q$, so that replacing all the probabilities by the average
probability $xk/k$ yields a lower value. Substituting
(\ref{eq:firstpart}) by (\ref{eq:boundonfirstpart}) and using the
expression (\ref{eq:optbound}) for $OPT$ then yields the desired
bound.
\begin{align*}
  & SW \geq (1 - 1/e) \Big(\sum_{q=1}^k \hat{v}(q) + \sum_{q=1}^k
    (\mathbf{E}_w[\check{w}(q)\ |\ \check{w}(q) \geq \hat{v}(q)] -
    \notag \\ 
  &\qquad \qquad\qquad\qquad\qquad
    \hat{v}(q))\mathbf{Pr}_w[\check{w}(q) \geq \hat{v}(q)]\Big) \\ 
  & = (1 - 1/e)OPT(f,g,k) .
\end{align*}

\end{proof}

Currently we have no non-trivial lower bound on the best approximation
factor achievable by a DSIC, IR, SBB mechanism, and we believe that
the approximation factor of $e/(e-1)$ achieved by our second mechanism
is not the best possible. For our first mechanism, it is rather easy
to see that the analysis of the approximation factor of our first
mechanism is tight, and that it is a direct extension of the median
mechanism of \cite{mcafee}, for which it was already shown in
\cite{bilateral} that it does not achieve an approximation factor
better than $2$: The authors show that $2$ is the best approximation
factor possible for any deterministic mechanism for which the choice
of $p$ does not depend on the buyer's distribution.

For the more general class of increasing valuation functions, an
approximation factor of $(2e-1) / (e-1) \approx 2.582$ to the optimal
social welfare is achieved by a mechanism of \cite{bilateral}: They
use a $e/(e-1)$-approximation mechanism for the single-item setting,
which yields a $(2e-1)/(e-1)$ approximation mechanism for the
multi-unit setting through a conversion theorem which they prove. We
note that their conversion theorem is more precisely presented for the
setting with a buyer and a seller who holds one \emph{divisible}
item. However, their proof straightforwardly carries over to the
multi-unit setting. It would be an interesting open challenge to
improve this currently best-known bound of $(2e-1)/(e-1)$ for general
increasing valuations.

\bibliographystyle{plain}
\bibliography{mbt}

\appendix

\section{Proof of the Characterisation for General
  Valuations}\label{sec:proof-char-gener} 

We denote the class of monotonically increasing submodular functions with 
domain $[k]$ by $\mathcal{S}_k$. We denote the class of monotonically 
increasing functions with domain $[k]$ by $\mathcal{I}_k$.

The definition below defines the multi-unit fixed price mechanisms as a direct 
revelation mechanism. From the point of view of providing a rigorous proof, this 
is more convenient to work with than the sequential posted price definition given 
in the main part of the paper.

\begin{definition}\label{def:char}
	Let $p \in \mathbb{R}_{\geq 0}$, let $S \subseteq [k]$, and let $\tau =(\tau_B, 
	\tau_S, \tau_{\cap})$ be a vector of three tie-breaking functions specified 
	below. The \emph{multi-unit fixed price mechanism} $\mathbb{M}_{p,S,\tau}$ 
	is the direct revelation mechanism that returns for a multi-unit bilateral trade 
	instance $(f,g,k)$ an outcome $\mathbb{M}_{p,S,\tau}(v,w) = 
	(q_B,q_S,\rho_B,\rho_S)$ on reported valuation functions $v$ and $w$, where
	\begin{itemize}
		\item $\tau_B(v) \subseteq \arg_q \max\{v(q) - qp : q \in S \cup \{0\}\}$ and 
		$\tau_B(v) \not= \varnothing$,
		\item $\tau_S(w) \subseteq \arg_q \max\{w(k-q) + qp : q \in S \cup \{0\}\}$ 
		and $\tau_S(w) \not= \varnothing$,
		\item $\tau_{\cap}(v,w)$ is a tie-breaking function that selects an element in 
		$\tau_B(v) \cap \tau_S(w)$ in case this intersection is non-empty,
		\item $q_B = k-q_S = \begin{cases} \min\{\max \tau_B(v), \max \tau_S(w)\} 
		& \text{ if } d_B \cap d_S = \varnothing , \\ \tau_{\cap}(v,w) & \text{ 
		otherwise. }\end{cases}$,
		\item $\rho_B = -\rho_S = q_Bp$.
	\end{itemize}
	Informally stated, the mechanism offers the buyer and seller a fixed unit price 
	$p$ and a set of quantities $S$. It then asks the buyer and seller which 
	quantity in $S \cup \{0\}$ they would like to trade when for each unit the buyer 
	would pay $p$ to the seller.
	The mechanism then makes the buyer and seller trade the minimum of these 
	two demanded numbers at a unit price of $p$. Typically the preferred quantity 
	is unique for both the buyer and the seller, but in case of indifferences the 
	buyer and seller will specify a set of multiple preferred quantities. In such 
	cases, the tie-breaking functions $\tau_B, \tau_S$ determine which quantities 
	among the sets of indifferences are considered for trade, and the tie-breaking 
	function $\tau_{\cap}$ is finally used to determine the traded quantity in case 
	the sets selected by $\tau_B$ and $\tau_S$ intersect. Otherwise, the minimum 
	of the maximum quantities of $\tau_B$ and $\tau_S$ is traded.
\end{definition}

It turns out that multi-unit fixed price mechanisms characterise the set of all 
DSIC, IR, and SBB mechanisms with respect to monotonically increasing 
submodular valuation functions. Moreover, with the additional restriction that 
$S$ is a singleton set, they characterise the set of all DSIC, IR, and SBB 
mechanisms with respect to monotonically valuation functions.

We first prove sufficiency.
\begin{theorem}\label{thm:suf}
	For all $p \in \mathbb{R}_{\geq 0}$ and $S \subseteq [k]$, the mechanism is 
	IR, SBB, and DSIC with respect to the class of monotonically increasing 
	submodular valuation functions. Moreover, if $|S| = 1$, then 
	$\mathbb{M}_{p,S,\tau}$ is IR, SBB, and DSIC with respect to the class of 
	monotonically increasing valuation functions.
\end{theorem}
\begin{proof}
	First we prove the statement for the class of monotonically increasing 
	submodular valuation functions.
	The SBB property holds trivially by definition of the mechanism, $\rho_B = 
	-\rho_S$.
	
	Let $v$ and $w$ be increasing submodular valuation functions of the buyer 
	and seller. Let $q_B$ and $\rho_B$ be the quantity given to the buyer and 
	payment made by the buyer under the outcome $\mathbb{M}_{p,S,\tau}(v,w)$. 
	If $\tau_B(v) \cap \tau_S(w)$ is non-empty, then the function $\tau_{\cap}$ 
	selects a utility maximizing quantity for both the buyer and seller, so IR 
	obviously holds in that case. If $\tau_B(v) \cap \tau_S(w) = \varnothing$, the 
	mechanism $\mathbb{M}_{p,S,\tau}$ is IR for the buyer: his utility is $v(q_B) - 
	\rho_B = v(q_B) - q_Bp = v(\min\{\max \tau_B(v),\max \tau_S(w)\}) - \min\{\max 
	\tau_B(v),\max \tau_S(w)\}p$. The value $\max \tau_B(v)$ is defined as a 
	utility-maximizing quantity in $S$ for the buyer, given that the buyer pays $p$ 
	for each unit. If the buyer's valuation function $v$ is submodular, getting any 
	quantity less than $\max \tau_B(v)$ at a price of $p$ per unit will yield the 
	buyer a non-negative utility. Therefore, the buyer's utility is non-negative. For 
	the seller, the argument to establish the IR property is similar: His utility is 
	$w(k - q_B) + \rho_B = w(k - q_B) + q_Bp = w(k - \min\{\max \tau_B(v),\max 
	\tau_S(w)\}) + \min\{\max \tau_B(v), \max \tau_S(w)\}p$. The value $\max 
	\tau_S(w)$ is defined as a utility maximizing quantity in $S$ for the seller to 
	give to the buyer, given that the seller receives a payment of $p$ for each 
	unit. As the buyer's valuation function $v$ is submodular, giving any quantity 
	less than $\max \tau_S(w)$ to the buyer at a price of $p$ per unit will yield the 
	seller a non-negative increase utility. Therefore, the seller's utility increase is 
	non-negative. 
	
	For the DSIC property, observe that if the mechanism sets $q_B \in 
	\tau_S(w)$, then the mechanism chooses the outcome that is the 
	utility-maximizing one for the seller among all outcomes in the range of the 
	mechanism. On the other hand, if $q_B \in \tau_B(v) \setminus \tau_S(w)$ then 
	the seller can only manipulate the outcome by misreporting a valuation that 
	causes $q_B$ to attain a smaller value, and hence in this case the mechanism 
	will select an outcome where a smaller quantity is traded against a price of 
	$p$ per unit. By increasingness and submodularity of the seller's valuation 
	function, this will result in a lower utility for the seller. Hence, it is a dominant 
	strategy for the seller to not misreport his valuation function. For the buyer, 
	the argument is similar: If the mechanism sets $q_B \in \tau_B(v)$, then the 
	mechanism chooses the outcome that is the utility-maximizing one for the 
	buyer among all outcomes in the range of the mechanism. On the other hand, 
	if $q_B \in \tau_S(w) \setminus \tau_B(v)$ then the buyer can only manipulate 
	the outcome by misreporting a valuation that causes $q_B$ to attain a smaller 
	value, and hence in this case the mechanism will select an outcome where a 
	smaller quantity is traded against a price of $p$ per unit. By increasingness 
	and submodularity of the buyer's valuation function, this will result in a lower 
	utility for the buyer. Hence, it is a dominant strategy for the buyer to not 
	misreport his valuation function.

	Next, we prove the statement for the larger class of monotonically increasing 
	valuation functions. Again, the SBB property holds trivially.
	
	As we now work under the assumption that $|S| = 1$, let $q$ be the quantity 
	such that $S = \{q\}$. Let $v$ and $w$ be increasing valuation functions for 
	the buyer and seller respectively. By definition of the mechanism and the 
	increasingness of the valuation functions, it holds that $\tau_B(v) \in 
	\{\{q\},\{0\},\{q,0\}\}$. Likewise, $\tau_S(w) \in \{\{q\},\{0\},\{q,0\}\}$. Therefore, 
	for both the buyer and seller, the traded quantity is $0$ or the unique positive 
	quantity $q$ in case he prefers trading $q$ units at least as much as trading 
	$0$ units. Hence the buyer and seller both experience a non-negative increase 
	in utility for the outcome decided by the mechanism. This establishes IR. For 
	DSIC, observe that if a positive quantity is traded in the selected outcome 
	under truthful reporting, then the only effect that misreporting can achieve is 
	that a quantity of $0$ at a price of $0$ is traded instead, which would leave 
	both the buyer and the seller with a $0$ increase in utility, hence this will not 
	increase either player's utility. If on the other hand a quantity of $0$ is traded 
	at a price of $0$, then $\{0\} = \tau_B(v), 0 \not\in \tau_S(w)$ or $\{0\} = 
	\tau_S(w), 0 \not\in \tau_B(v)$ or $0 \in \tau_B(v), 0 \in \tau_S(w)$. In the first 
	case, clearly the buyer is not incentivised to manipulate the mechanism into 
	producing the alternative outcome where $q$ units are traded, and the seller 
	is unable to manipulate the mechanism into producing that outcome as it 
	selects the minimum of $\tau_S(w)$ and $\tau_B(w)$, where the latter equals 
	$\{0\}$ regardless of the sellers report. For the second case, symmetric 
	reasoning can be applied to conclude that none of the two agents are 
	incentivised to misreport. For the third case, it trivially holds that none of the 
	agents are incentivised to manipulate the mechanism into trading $q$ instead 
	of $0$ units. This establishes DSIC.
\end{proof}

Next, we show necessity, i.e., all DSIC, IR, and SBB direct revelation 
mechanisms are multi-unit fixed price mechanisms.
\begin{theorem}\label{thm:char}
	Let $\mathbb{M}$ be a multi-unit bilateral trade mechanism that is IR, SBB, 
	and DSIC with respect to the class of monotonically increasing submodular 
	valuation functions. Then, there exist $p \in \mathbb{R}_{\geq 0}$, $S 
	\subseteq [k]$, and $\tau$ such that $\mathbb{M} = \mathbb{M}_{p,S,\tau}$. 
	Moreover, if $\mathbb{M}$ is also IR, SBB, and DSIC with respect to the 
	bigger class of monotonically increasing valuation functions, then $|S| = 1$.
\end{theorem}

We divide this proof up into several lemmas. We start by proving the theorem for 
the smaller class of monotonically increasing submodular valuation functions. 
First, we show that whenever $\mathbb{M}$ trades the same number of items for 
two distinct pairs of valuation functions, then it must charge the same payments. 
Second, we extend this by showing that whenever the mechanism trades 
distinct numbers of items for any two distinct pairs of valuation functions, then 
the mechanism must charge the same price proportional to the number of items 
traded. It follows that we may associate to $\mathbb{M}$ a unit price $p$ such 
that the payment from the buyer to the seller is always $q_Bp$, where $q_B$ is 
the traded quantity.
Lastly, we show that there is a set $S$ such that the range of quantities that the 
seller may let the mechanism trade from (by means of reporting a valuation 
function to the mechanism), is equal to $(S \cup \{0\}) \cap [\arg_q\max\{v(q) - pq 
: q \in S \cup \{0\}\}]$. By the fact that the valuation functions are increasing and 
submodular, and by the fact that $\mathbb{M}$ is DSIC, it follows that truthful 
reporting of the seller will result in the mechanism trading 
\begin{equation*}
	\arg_q \max\{w(k-q) + pq : q \in (S \cup \{0\}) \cap [\arg_q\max\{v(q) - pq : q \in 
	S \cup \{0\}\}]\} 
\end{equation*}
units. This expression is equal to $\min\{\max d_B, \max d_S\}$ if $d_B \cap 
d_S = \varnothing$ (where $d_B$ and $d_S$ are defined as in Definition 
\ref{def:char}) , and otherwise it is a set from which an arbitrary quantity 
$\tau(v,w)$ may selected. This implies that $\mathbb{M} = 
\mathbb{M}_{p,S,\tau}$ for the appropriate choices of $p$, $S$, and $\tau$.

With respect to the larger class of monotonically increasing valuation functions, 
the set of DSIC, IR, and SBB mechanisms must be smaller. We prove for this 
class that whenever the set $S$ consists of more than one quantity, then there 
must be a pair of valuation functions in which either the buyer or seller is better 
off by not truthfully reporting his valuation function.

We now proceed by stating and proving formally the claims sketched above.
In the proofs of the claims below, we use the following terminology and 
notation. 
For ease of exposition, we denote from now on an outcome by a pair $(q, p)$ 
where $q$ is the traded number of units (i.e., the quantity that the buyer gets 
assigned) and $p$ is the payment of the buyer, which is equal to the negated 
payment of the seller by the SBB property.
For a reported valuation $v$ of the buyer, let $M_v = \{o \in \mathcal{O}\ |\ 
\exists w : \mathbb{M}(v,w) = o\}$ be the menu of outcomes offered to the seller 
when the buyer reports $v$. That is, when the buyer reports $v$, the seller can 
select one of the outcomes $o$ in $M_v$ by reporting (not necessarily truthfully) 
some valuation in reply to $v$. Likewise, we let $N_w = \{o \in \mathcal{O}\ |\ 
\exists v : \mathbb{M}(v,w) = o\}$ be the menu of outcomes offered to the buyer 
when the seller reports $v$. We let $M = \bigcup_v M_v = \bigcup_w N_w$ be 
the set of all outcomes that the mechanism can produce, and we let $S$ be the 
projection of $M$ on the quantity obtainable by the buyer (i.e., the set $S$ 
consists of all quantities that the mechanism can possibly trade). 

The next lemma shows that there is a unique payment that the mechanism 
charges for every quantity in $S$, which implies that $\mathbb{M}$ consists of 
at most $k$ outcomes.

\begin{lemma}\label{lem:uniqueprice}
	Let $\mathbb{M}$ be IR, SBB, and DSIC with respect to the class of valuation 
	functions $\mathcal{C}$, where $\mathcal{C}$ is either $\mathcal{S}_k$ or 
	$\mathcal{I}_k$.
	Let $(v,w)$ and $(v',w')$ be two pairs in $\mathcal{C}^2$. Let 
	$\mathbb{M}(v,w) = (q_B,\rho_B)$ and $\mathbb{M}(v',w') = (q_B',\rho_B')$. If 
	$q_B = q_B'$, then $\rho_B = \rho_B'$.
\end{lemma}
\begin{proof}
	As $\mathbb{M}$ is DSIC, it is immediate that for every $v''$ and for every 
	quantity $q$ it holds that there are no two distinct payments $p, p'$ such that 
	$(q,p)$ and $(q, p')$ are both in $M_{v''}$. Also, let $(q,p)$ and $(q',p')$ be in 
	$M_{v''}$, where $q < q'$. Then $p \leq p'$, as otherwise there are valuations 
	of the seller where misreporting results in trading less items at a higher price, 
	which would violate the DSIC property. 
	
	Let $(v,w)$ and $(v',w')$ be as in the statement of the lemma, i.e., such that 
	$q_B = q_B'$. 
	When the buyer reports $v$ and the seller reports $w$, by assumption $(q_B, 
	\rho_B)$ is the outcome, so $(q_B,\rho_B) \in M_v$.
	Define the valuation function $w^*$ as the function that grows linearly, 
	extremely steeply
	up to the quantity $k-q_B = k- q_B'$, and grows extremely slowly at a rate of 
	$\epsilon > 0$ from $k - q_B$ onward. 
	We define the function $v^*$ similarly: It grows at an extremely high rate up to 
	the quantity $q_B = q_B'$ and grows at extremely slow rate $\epsilon$ from 
	$q_B'$ onward.
	
	We first consider a deviation by the seller from $(v,w)$ to $(v,w^*)$.
	Let $(q,p) \in M_v$ be the outcome of the mechanism on report $(v,w^*)$.
	As we have chosen the valuation $w^*$ to be sufficiently steep up to 
	$k-q_B$ items, IR would be violated for a seller with valuation $w^*$ if more 
	than $q_B$ items are traded. 
	
	Suppose now that upon report $(v,w^*)$ the mechanism trades strictly less 
	than $q_B$ items. i.e., $q < q_B$. We prove that then, $p = \rho_B$: If we 
	would assume $p < \rho_B$, then a seller with valuation $w^*$ would 
	misreport $w$, as in terms of valuation he is practically indifferent between 
	trading $q$ and $q_B$ items ($\epsilon$ needs to be chosen small enough 
	for this), and his received payment would increase from $p$ to $\rho_B$. If on 
	the other hand we would assume that $p > \rho_B$, then a seller with 
	valuation $w$ would misreport $w^*$ as he would retain more items, and 
	receive a higher payment. Thus $p = \rho_B$. 
	
	By entirely analogous reasoning, when $(v',w^*)$ is reported to the 
	mechanism, the mechanism also trades $q_B' = q_B$ items or less. Let $q' < 
	q_B$ be the number of traded items under $(v',w^*)$. The payment is equal 
	to $\rho_B'$, and if less than $q_B$ items are traded, then $w'(k-q') = 
	w'(k-q_B)$.
	
	Next, we define from $w^{*}$ a valuation function $w^{**}$ for which it holds 
	that under $(v,w^{**})$ and $(v',w^{**})$ the same number of items is traded at 
	prices $\rho_B$ and $\rho_B'$ respectively. We do this as follows: If $q' = q$ 
	then we simply let $w^{**}$ be $w^*$. Otherwise, if $q' \not= q$, assume 
	without loss of generality that $q' > q$ and let $\bar{w}^*$ be the valuation 
	function that grows extremely steeply up to $k-q$ units and increases 
	extremely slowly after $k-q$ units. By considering the deviation by the seller 
	from profile $(v,w^*)$ to profile $(v,\bar{w}^{*})$, we see that under 
	$(v,\bar{w}^*)$ at most $q'$ units are traded and the payment is still equal to 
	$q_B$. Likewise, under $(v',\bar{w}^*)$ at most $q'$ items are traded and the 
	payment is $q_B'$. Thus, the minimum number of traded items among the 
	pair of strategy profiles $(v,\bar{w}^{*})$ and $(v',\bar{w}^{*})$ is larger than the 
	minimum number of traded items among the pair $(v,\bar{w}^{*})$ and 
	$(v',\bar{w}^{*})$. Repeating this operation will thus eventually yield a strategy 
	profile $w^{**}$ such that under $(v,w^{**})$ and $(v',w^{**})$ the same number 
	of items is traded at prices $\rho_B$ and $\rho_B'$ respectively.
	
	Now, if we would suppose for contradiction that $\rho_B \not= \rho_B'$, then 
	we may assume without loss of generality that $\rho_B < \rho_B'$. When the 
	seller reports $w^{**}$, a buyer with valuation function $v$ would now be 
	incentivised to report the valuation function $v'$ instead of $v$, since then his 
	payment decreases, and he still receives the same number of items. This is a 
	contradiction to the DSIC property. Therefore, $\rho_B = \rho_B'$ which 
	proves our claim.
\end{proof}

By Lemma \ref{lem:uniqueprice} there is a unique payment for each quantity $q 
\in S$, and we denote this payment by $p(q)$. 
The next lemma extends the previous lemma by stating essentially that 
payments must grow linearly with the number of allocated items, when one of 
the players changes his reported valuation.
\begin{lemma}\label{lem:linearprice}
	Let $\mathbb{M}$ be IR, SBB, and DSIC with respect to the class 
	$\mathcal{C}$, where $\mathcal{C}$ is either $\mathcal{S}_k$ or 
	$\mathcal{I}_k$. 
	Let $(v,w)$ and $(v',w')$ be two pairs in $\mathcal{C}^2$. Let 
	$\mathbb{M}(v,w) = (q_B,q_S,\rho_B,\rho_S)$ and $\mathbb{M}(v',w') = 
	(q_B',q_S',\rho_B',\rho_S')$. If $q_B > 0$, then $\rho_B' = (q_B'/q_B)\rho_B$.
\end{lemma}
\begin{proof}
	First we show that $p(\cdot)$ is a non-decreasing function. Suppose that this 
	is not true, and assume that $(q, p(q))$ and $(q', p(q))$ is the pair of outcomes 
	in $M$ such that (i) $q'>q$ and $p(q) > p(q')$, (ii) there is no $q'' \in S$ such 
	that $q'>q''>q$ and (iii) $q$ is minimal. Let $(v,w)$ and $(v',w')$ be two 
	valuation profiles that result in these two respective outcomes $(q, p(q))$ and 
	$(q', p(q))$. Let $w^*$ be a valuation function that increases linearly at an 
	extremely high rate up to $k - q'$ and increses extremely slowly afterward. 
	We now see that when $(v,w^*)$ is reported, $q'$ units or less are traded due 
	to IR, and in fact at most $q$ units are traded due to DSIC (because, if $q'$ 
	units were traded, a seller with valuation $w^*$ would misreport $w$ to trade 
	less units for more money), and no less than $q''$ units are traded due to 
	DSIC, where $q'' \leq q$ is the least quantity such that $p(q'') = p(q)$. 
	(Otherwise the seller with valuation $w^*$ could misreport $w$ and trade 
	more units at almost the same valuation, for significantly less money. We use 
	here that $w^*$ increases sufficiently slowly on  the interval $[k-q', k]$).
	We also observe that when $(v',w^*)$ is reported, (i) $q'$ units or less are 
	traded due to IR, (ii) the traded quantity cannot be any quantity with a higher 
	payment than $p(q')$ since otherwise a seller with valuation $w'$ would 
	misreport $w^*$ if the buyer reports $v'$, and (iii) the traded quantity cannot 
	be any quantity with a lower payment than $p(q')$ since otherwise a seller 
	with valuation $w^*$ would misreport $w'$. Thus, exactly $q'$ units are 
	traded when $(v',w^*)$ is reported. We conclude that $(q,p(q))$ and $(q',p(q'))$ 
	are both in $N_{w^*}$, and therefore a buyer with valuation $v$ would have an 
	incentive to misreport $v'$ if the seller reports $w^*$, which violates DSIC 
	and yields a contradiction. We conclude that the payment function $p(\cdot)$ 
	is non-decreasing.
	
	Also, note that $0 \in S$ and $p(0) = 0$ as otherwise the IR property would be 
	violated for a buyer whose valuation is identically $0$ and a seller whose 
	valuation is strictly increasing: When such a valuation profile is reported, the 
	seller's valuation implies that no positive number of items can be traded for 
	payment 0; the buyer's valuation implies that no positive number of items can 
	be traded for a positive payment; so $0$ items must be traded for a payment 
	that is not positive (due to IR of the buyer) and not negative (due to IR of the 
	seller). 

	The claim of this lemma is equivalent to the claim that there exists a unique 
	value $p$ such that $p(q) = pq$ for all $q \in S$. We will demonstrate this by 
	means of contradiction: Suppose that there is no such value $p$. Let $q$ be 
	the lowest quantity in $S$ such that $p(q'') \not= p(q)q''/q$ for all $q'' \in S, 
	q'' < q$. Let $q'$ be the highest quantity in $S$ such that $p(q'') = p(q')q''/q'$ 
	for all $q'' \in S, q'' < q'$. Note that there is no $q'' \in S$ such that $q' < q'' 
	< q$, and that $p(\cdot)$ behaves linearly up to $q'$, and that $q$ essentially 
	serves as the least witness for the non-linearity of $p(\cdot)$.
	
	We distinguish two cases: The case where $p(q)q'/q > p(q')$, and the case 
	where $p(q)q'/q < p(q')$. First let us assume that $p(q)q'/q > p(q')$. We will 
	derive a contradiction by constructing valuation functions $(v^*,w^*)$ for the 
	buyer and seller such that the following properties are satisfied: (i) outcomes 
	$(q, p(q))$ and $(q',p(q'))$ are in $M_{v^*}$ and $N_{w^*}$; (ii) the buyer with 
	valuation $v^*$ strictly prefers outcome $(q',p(q'))$ over all outcomes in 
	$N_{w^*} \setminus \{(q',p(q'))\}$; and (iii) the seller with valuation $w^*$ 
	strictly prefers outcome $(q, p(q))$ over all outcomes in $M_{v^*} \setminus 
	\{(q,p(q))\}$. This is a contradiction because (ii) requires that the mechanism 
	outputs $(q',p(q'))$ (otherwise the buyer with valuation $v^*$ would be 
	incentivised to misreport so that $(q',p(q'))$ is output) and (iii) requires that the 
	mechanism outputs $(q,p(q))$ (otherwise the seller with valuation $w^*$ would 
	be incentivised to misreport so that $(q,p(q))$ is output). 
	
	Therefore, we will now define the appropriate valuations $v^*$ and $w^*$. Let 
	$v^*$ be a valuation function that grows linearly at an extremely high rate up 
	to quantity $q'$ and increases extremely slowly afterward. This causes all 
	outcomes where a positive quantity is traded a positive utility for the buyer 
	with valuation $v^*$, moreover, the maximum utility for such a buyer is 
	achieved at outcome $(q',p(q'))$. This already establishes property (ii). To see 
	that property (i) holds, let $w$ be any seller's valuation so that $(q',p(q')) \in 
	N_w$ (which must exist because $q' \in S$). By the definition of $v^*$, the 
	mechanism selects outcome $(q',p(q'))$ on report $(v^*,w)$ and therefore 
	$(q',p(q')) \in M_{v^*}$. Now, consider any valuation profile $(v,w)$ that results 
	in outcome $(q,p(q))$, so that $(q,p(q)) \in M_v$. Let $w'$ be the valuation 
	function that grows linearly at an extremely high rate up to quantity $k-q$ and 
	after that point grows linearly at a rate of $p(q)/q-\epsilon/q$ up to quantity 
	$k$. The initial increase up to point $k-q$ is so steep that the seller can never 
	experience a utility above $w'(k)$ when any quantity higher than $q$ is traded. 
	The value $\epsilon > 0$ is chosen to be so small that the only outcome at 
	which a seller with valuation $w'$ has a positive utility is $(q,p(q))$. Therefore, 
	upon report $(v,w')$ the mechanism outputs $(q,p(q))$ and we may infer that 
	$N_{w'} = \{(q,p(q)),(0,0)\}$, from which it follows that $(q,p(q)) \in M_{v^*}$, as 
	$(q,p(q))$ must be the selected outcome upon report $(v^*,w)$ (by the DSIC 
	property). This establishes property (i) for $v^*$.
	
	For valuation function $w^*$, let $w^*$ increase linearly at an extremely high 
	rate up to quantity $k-q$, and increase extremely slowly afterwards. Clearly, 
	the seller with valuation $w^*$ prefers the outcome $(q,p(q))$ among all 
	outcomes in $M$, which establishes property (iii). Let $(v,w)$ be any report 
	upon which the mechanism outputs $(q,p(q))$, so that $(q,p(q)) \in M_v$. Then 
	$(q,p(q))$ is also output upon report $(v,w^*)$ which establishes that $(q,p(q)) 
	\in N_{w^*}$. Next, let $(v,w)$ be any report upon which the mechanism 
	outputs $(q',p(q'))$, so that $(q',p(q')) \in N_w$. Let $v'$ be a function that 
	increases at rate $p(q')/q' + \epsilon/q'$ for sufficiently small $\epsilon > 0$ 
	up to quantity $q'$, and increases extremely slowly afterward. Then 
	$(q',p(q'))$ is output when $(v',w)$ is reported, so that $(q',p(q')) \in M_{v'}$. 
	Moreover, trading any quantity higher than $q'$ would yield a negative utility 
	for a buyer with valuation $v'$ (because $\epsilon$ is extremely small). 
	Therefore, for $q'' > q'$ it holds that $(q'' , p(q'')) \not\in M_{v'}$ and in 
	particular $(q, p(q)) \not\in M_{v'}$. Thus, when $(v',w^*)$ is reported, the 
	outcome $(q',p(q'))$ is output by the mechanism, and this establishes that 
	$(q',p(q')) \in N_{w^*}$. Note that here we need that $p(q') > 0$, which is the 
	case as the outcome the mechanism returns on $(v',w')$ is IR by assumption 
	and the valuation functions are monotonically increasing. This completes the 
	proof for the case where $p(q)q'/q > p(q')$.

	For the case where $p(q)q'/q < p(q')$ we proceed in a similar fashion: Again, 
	we will derive a contradiction by constructing valuation functions $(v^*,w^*)$ 
	for the buyer and seller such that (i) outcomes $(q, p(q))$ and $(q',p(q'))$ are in 
	$M_{v^*}$ and $N_{w^*}$; (ii) the buyer with valuation $v^*$ strictly prefers 
	outcome $(q,p(q))$ over all options in $N_{w^*} \setminus \{(q,p(q))\}$; and (iii) 
	the seller with valuation $w^*$ strictly prefers outcome $(q', p(q'))$ over all 
	options in $M_{v^*} \setminus \{(q,p(q))\}$. This is a contradiction because (ii) 
	requires that the mechanism outputs $(q,p(q))$ (otherwise the buyer with 
	valuation $v^*$ would be incentivised to misreport so that $(q,p(q))$ is output) 
	and (iii) requires that the mechanism outputs $(q',p(q'))$ (otherwise the seller 
	with valuation $w^*$ would be incentivised to misreport so that $(q',p(q'))$ is 
	output). Note that the difference with the previous case is that here we 
	construct $v^*$ such that the higher of the two quantities $q$ and $q'$ is 
	preferred, instead of the lower one. Likewise, $w^*$ is now constructed such 
	that the lower of the two quantities is preferred instead of the higher one.
	
	We start in this case with the construction of $w^*$. Let $w^*$ be a valuation 
	function that increases linearly at an extremely high rate up to quantity $k-q$. 
	From $k-q$ to $k-q'$, valuation $w^*$ increases by an amount of $p(q) - p(q') 
	+ \epsilon$, where $\epsilon > 0$ is sufficiently small, and $w^*$ increases 
	extremely slowly from $k-q'$ onward. The increase in valuation from 
	quantities $k-q$ to $k-q'$ is slightly higher than the amount by which the 
	payment changes among the quantities $q$ and $q'$, this causes the seller 
	with valuation $w^*$ to encounter a slightly lower (but positive) increase in 
	utility when quantity $q$ is traded instead of quantity $q'$. Moreover, among 
	all quantities in $S$ up to $q'$, the maximum utility for a seller with valuation 
	$w^*$ is achieved at quantity $q'$, which already establishes property (iii). 
	Lastly, note that due to the extreme steepness of $w^*$ up to $k-q$, the 
	utility of the seller is lower than $w^*(k)$ when any quantity higher than $q$ is 
	traded, so the mechanism will never do so by the IR constraint. It remains to 
	establish property (i). Let $(v,w)$ be any report where the mechanism selects 
	outcome $(q',p(q'))$. It follows by DSIC that outcome $(q',p(q'))$ will also be 
	selected on report $(v,w^*)$, so that $(q',p(q')) \in N_{w^*}$. Next, let $(v,w)$ 
	be any report where the mechanism selects outcome $(q,p(q))$. Let $w'$ be a 
	valuation function that increases extremely steeply up to $k-q$ and increases 
	extremely slowly afterwards, so that the report $(v,w')$ results in $(q,p(q))$ 
	and hence $(q,p(q)) \in N_{w'}$, and because of IR we also infer that 
	$(q'',p(q'')) \not\in M_{w'}$ when $q'' > q$. Let $v'$ be a valuation function 
	that increases linearly up to quantity $q$ and increases extremely slowly 
	afterwards, where $v(q) = p(q) + \epsilon$, and $\epsilon > 0$ is sufficiently 
	small. Note that trading a positive quantity lower than $q$ would result in a 
	negative utility for a buyer with valuation $v'$, so that such outcomes are not 
	in $M_{v'}$. Therefore, when $(v',w')$ is reported the outcome selected by the 
	mechanism must be $(q,p(q))$, which shows that $(q,p(q)) \in M_{v'}$. It 
	follows now that the selected outcome upon report $(v',w^*)$ must be 
	$(q,p(q))$ which yields $(q,p(q)) \in N_{w^*}$ and establishes property (i) for 
	$w^*$.
	
	Lastly, we design $v^*$. Let $v^*$ simply increase extremely steeply up to 
	the quantity $q$, and increase extremely slowly  afterwards, so that a buyer 
	with valuation $v^*$ experiences positive utility for all outcomes in $M$, and 
	maximum utility when outcome $(q,p(q))$ is selected. This straightforwardly 
	establishes property (ii). For property (i), let $(v,w)$ be any profile where 
	outcome $(q,p(q))$ results, so that $(q,p(q)) \in N_w$. By DSIC, outcome 
	$(q,p(q))$ is also selected when $(v^*,w)$ is reported, so $(q,p(q)) \in 
	M_{v^*}$. Next, let $(v,w)$ be any profile where outcome $(q',p(q'))$ results, 
	so $(q',p(q')) \in M_v$. Let $w'$ be a function that increases extremely steeply 
	up to quantity $k-q'$, and increases extremely slowly afterwards, so that 
	trading any quantity higher than $q'$ would result in a decrease in utility for a 
	seller with valuation $w'$ (hence the mechanism cannot trade such quantities 
	when $w'$ is reported, by the IR property), and the maximum increase in 
	utility is achieved when $(q',p(q'))$ is chosen. Therefore reporting $(v,w')$ 
	results in outcome $(q',p(q'))$, thus $(q',p(q')) \in N_{w'}$ and $(q'',p(q'')) 
	\not\in N_{w'}$ for all $q'' > q'$. Therefore, when $(v^*,w')$ is reported, 
	outcome $(q',p(q'))$ is selected, which establishes property (i) for $v^*$ and 
	completes the proof for the case $p(q)q'/q > p(q')$.
	
\end{proof}

Let $\mathbb{M}$ be IR, SBB, and DSIC with respect to the class of 
monotonically increasing submodular valuation functions. From the above it 
follows that for a mechanism that is IR, SBB, and DSIC with respect to 
$\mathcal{S}_k$ or $\mathcal{I}_k$, there exists a price $p \in \mathbb{R}_{\geq 
0}$ such that for all pairs $(v,w)$ of monotonically increasing submodular 
valuation functions, the payment charged to the buyer is $q_Bp$ (and the 
payment charged to the seller is $-q_Bm$ by SBB). We will refer to $p$ as the 
\emph{unit price}.

The above corollary establishes the needed properties on the payments of the 
mechanism.
The remaining lemmas use Lemma \ref{lem:linearprice} by implicitly assuming the 
existence of the unit price $p$ in their statement, and they characterise the 
quantities $S$ tradable by the mechanism and the quantities that appear in the 
menus $M_v$ and $N_w$. 
The next lemma states that the utility maximizing outcome in $S$ for a buyer 
with any valuation function $v$ is always in $M_v$.
\begin{lemma}
	If $\mathbb{M}$ is SBB, IR, and DSIC with respect to $\mathcal{S}_k$, and 
	suppose that unit price $p$ is positive. Then, for all $v \in \mathcal{S}_k$ it 
	holds that $(q,p(q)) \in M_v$ for the lowest $q$ in the set $\arg_q \max\{v(q) - 
	p(q) : q \in S \}$.
\end{lemma}
\begin{proof}
	Let $q$ be the lowest quantity in $\arg_{q} \max\{v(q) - p(q) : q \in S \}$. 
	Let $(v',w')$ be any report that results in outcome $(q, p(q))$, so that $(q,p(q)) 
	\in N_{v'}$. Let $w^*$ be a valuation function that increases extremely steeply 
	up to the quantity $k-q$ and increases extremely slowly afterwards. Observe 
	that by our assumption that $p > 0$, a seller with valuation $w^*$ strongly 
	prefers outcome $(q,p(q))$ over all other outcomes in $S$, and trading any 
	quantity larger than $q$ would violate IR. by DSIC, outcome $(q,p(q))$ is thus 
	selected when $(v',w^*)$ is reported, hence $(q,p(q)) \in N_{w^*}$ and 
	$(q',p(q')) \not\in N_{w^*}$ for all $q' > q$. So when $(v,w^*)$ is reported, an 
	outcome is selected from $N_{w^*}$ that maximises the utility of the buyer 
	with valuation $v$, and this outcome is $(q,p(q))$.
	
\end{proof}

The following lemma strengthens the previous.
\begin{lemma}\label{lem:buyermenu}
	Suppose $\mathbb{M}$ is SBB, IR, and DSIC with respect to $\mathcal{S}_k$ 
	and suppose that unit price $p$ is positive. Let $v \in \mathcal{S}_k$ and let 
	$q \leq \min \arg_{q'} \max \{v(q') - p(q') : q' \in S\}$ be a quantity not 
	exceeding the least utility-maximizing quantity for a buyer with valuation $v$. 
	It holds that $\arg_{q'} \max\{v(q') - p(q') : q' \in S, q' \leq q \}$ is the singleton 
	set containing the quantity $q' = \max S \cap [q]$, and that $q \in M_v$.
\end{lemma}
\begin{proof}
	Note that the existence of the unit price $p$ implies that the utility function of 
	the buyer is a submodular function of the traded quantity. Therefore, the utility 
	function for a buyer with valuation $v$ is increasing up to the least 
	utility-maximizing outcome in $S$, after which it stays constant up to the 
	highest utility-maximizing outcome in $S$, after which it starts decreasing. 
	Let $q$ be any quantity less than or equal to the least utility-maximizing 
	quantity, i.e., less than $\min \arg \max \{v(q') - p(q') : q' \in S\}$. Then, the 
	utility-maximizing quantity $q'$ in $S \cap [q]$ for a buyer with valuation $v$ 
	is $\max S \cap [q]$. It remains to prove that $(q',p(q'))$ is in $M_v$. Let 
	$(v',w')$ be any report resulting in outcome $(q',p(q'))$, so that $(q',p(q')) \in 
	N_{v'}$. Let $w''$ be any function increasing extremely steeply up to quantity 
	$k-q'$, after which it increases extremely slowly. Then $(q',p(q'))$ is the result 
	of report $(v',w'')$, and note that it is not IR to trade a quantity exceeding 
	$q'$ when $w''$ is reported, so $(q'',p(q'')) \not\in N_{w''}$ for $q'' > q'$, and 
	$(q',p(q')) \in N_{w''}$. Therefore, when $(v,w'')$ is reported, a quantity of $q'$ 
	is traded, and no higher quantity. (Note that we use positivity of $p$ here.) 
	Thus, $(q',p(q'))$ is in $M_{v}$, which proves the claim.
\end{proof}

The above lemma shows that for a mechanism $\mathbb{M}$ that is SBB, IR, 
and DSIC with respect to $\mathcal{S}_k$, if $p > 0$, then for any $v \in 
\mathcal{S}_k$, the menu $M_v$ that the buyer presents to the seller includes 
the outcomes $(q,p(q))$ such that $q$ is in the subset of $S$ obtained by 
truncating $S$ at the buyer's least-quantity utility-maximizing outcome. 

We can prove the following symmetric lemma for the seller.
\begin{lemma}\label{lem:sellermenu}
	Suppose $\mathbb{M}$ is SBB, IR, and DSIC with respect to $\mathcal{S}_k$ 
	and suppose that the unit price $p$ is positive. Let $w \in \mathcal{S}_k$ and 
	let $q \leq \min \arg_{q'} \max \{w(k-q') + p(q') : q' \in S\}$ be a quantity not 
	exceeding the least utility-maximizing quantity for a seller with valuation $w$. 
	It holds that $\arg_{q'} \max\{w(k-q') + p(q') : q' \in S, q' \leq q \}$ is the 
	singleton set containing quantity $q' = \max S \cap [q]$, and that $q \in N_w$.
\end{lemma}
\begin{proof}
	Note that the existence of the unit price $p$ implies that the increase in utility 
	of the seller is a submodular function of the traded quantity $q$. Therefore, 
	the function for a buyer with valuation $v$ is increasing up to the least 
	utility-maximizing outcome in $S$, after which it stays constant up to the 
	highest utility-maximizing outcome in $S$, after which it starts decreasing. 
	Let $q$ be any quantity less than or equal to the least utility-maximizing 
	quantity, i.e., less than $\min \arg \max \{w(k-q') + p(q') : q' \in S\}$. Then, the 
	utility-maximizing quantity $q'$ in $S \cap [q]$ for a seller with valuation $w$ 
	is $\max S \cap [q]$. It remains to prove that $(q',p(q'))$ is in $N_{w}$. Let 
	$(v',w')$ be any report resulting in outcome $(q',p(q'))$, so that $(q',p(q')) \in 
	N_{w'}$. Let $v''$ be any function increasing at a rate $p+\epsilon$ up to 
	quantity $q'$, for a sufficiently small $\epsilon > 0$, after which it increases 
	extremely slowly. Then $(q',p(q'))$ is the result of report $(v'',w')$, and note 
	that it is not IR to trade a quantity exceeding $q'$ when $v''$ is reported, so 
	$(q'',p(q'')) \not\in M_{v''}$ for $q'' > q'$, and $(q',p(q')) \in N_{v''}$. 
	Therefore, when $(v'',w)$ is reported, a quantity of $q'$ is traded, and no 
	higher quantity. Thus, $(q',p(q'))$ is in $M_{w}$, which proves the claim.
\end{proof}

The above lemma shows that for any unit price mechanism $\mathbb{M}$ that is 
SBB, IR, and DSIC with respect to $\mathcal{S}_k$, for any $w \in 
\mathcal{S}_k$, the menu $M_w$ that the seller presents to the buyer includes 
the outcomes $(q,p(q))$ such that $q$ is in the subset of $S$ obtained by 
truncating $S$ at the buyer's least-quantity utility-maximizing outcome. 

The last two lemmas combined imply that the menu of buyer consist of the 
outcomes $(q,p(q))$ in $S$ where the quantity does not exceed the least 
utility-maximizing outcome, plus an additional arbitrary subset of 
utility-maximizing outcomes; and the same holds for the seller. We will show that 
next.

\begin{lemma}\label{lem:menustructure}
	Suppose $\mathbb{M}$ is SBB, IR, and DSIC with respect to $\mathcal{S}_k$ 
	and suppose that the unit price $p$ is positive. Let $v,w \in \mathcal{S}_k$. 
	Let $q = \min \arg_{q''} \max \{v(q'') - p(q'') : q'' \in S\}$ be the least utility 
	maximizing quantity for the buyer with valuation $v$, then $M_v = \{(q'',p(q'')) 
	: q'' \in S \cap [q]\} \cup T$ where $T \subseteq \arg_{q''} \max \{v(q'') - p(q'') 
	: q'' \in S\}$. Similarly let $q' = \min \arg_{q''} \max \{w(k - q'') + p(q'') : q'' \in 
	S\}$ be the least utility maximizing quantity for the seller with valuation $w$, 
	then $N_w = \{(q'',p(q'')) : q'' \in S \cap [q']\} \cup T'$ where $T' \subseteq 
	\arg_{q''} \max \{w(k - q'') + p(q'') : q'' \in S\}$.
\end{lemma}
\begin{proof}
	Lemmas \ref{lem:buyermenu} and \ref{lem:sellermenu} show that $M_v 
	\supseteq \{(q'',p(q'')) : q'' \in S \cap [q]\}$ and $N_w \supseteq \{(q'',p(q'')) : 
	q'' \in S \cap [q']\}$. Let $\hat{q} \in S$ such that $\hat{q} > \max \arg_{q''} 
	\max \{v(q'') - p(q'') : q'' \in S\}$ and let $\check{q} \in S$ such that $\check{q} 
	> \max \arg_{q''} \max \{w(k - q'') + p(q'') : q'' \in S\}$. It suffices to show that 
	$(\hat{q},p(\hat{q})) \not\in M_v$ and that $(\check{q},p(\check{q})) \not\in N_w$.
	
	Suppose $(\hat{q},p(\hat{q})) \in M_v$. Let $w^*$ be a valuation function 
	increasing extremely steeply up to quantity $k-\hat{q}$, and increases 
	extremely slowly afterwards. By DSIC, the outcome $(\hat{q},p(\hat{q}))$ is 
	selected on report $(v,w^*)$, where we use that $p > 0$. However, by Lemma 
	\ref{lem:sellermenu} it holds that $(q,p(q)) \in N_{w^*}$, so that it also must 
	hold by the DSIC property that outcome $(q,p(q))$ is selected, which is a 
	contradiction.
	
	Suppose $(\check{q},p(\check{q})) \in N_w$. Let $v^*$ be a valuation 
	increasing at extremely high rate up to quantity $\check{q}$, that increases 
	extremely slowly afterwards. By DSIC, the outcome $(\check{q},p(\check{q}))$ 
	is selected on report $(v^*,w)$. However, by Lemma \ref{lem:buyermenu} it 
	holds that $(q',p(q')) \in N_{v^*}$, so that it also must hold by the DSIC 
	property that outcome $(q,p(q))$ is selected, which is a contradiction.
\end{proof}

We are now finally ready to prove the necessity-part of our characterisation of 
IR, DSIC, SBB multi-unit bilateral trade mechanisms.
\begin{proof}[Proof of Theorem \ref{thm:char}]
	By Lemma \ref{lem:uniqueprice}, for all $q \in S$ there is a price $p(q)$ such 
	that a payment of $p(q)$ is charged whenever $q$ units are traded. By Lemma 
	\ref{lem:linearprice}, there is a unit price $p$ such that $p(q) = p \cdot q$ for 
	all $q \in S$. This establishes already that the payment function of any IR, 
	DSIC, and SBB mechanism is in accordance with Definition \ref{def:char}, 
	hence it remains to establish that the traded quantity is also prescribed by 
	Definition \ref{def:char}.
	
	First we consider the special case $p = 0$. By increasingness of the valuation 
	function of the seller, it follows that the mechanism can only trade a quantity 
	of $0$ units in order to satisfy IR. Subsequently it follows by IR and SBB that 
	the mechanism is required to charge a payment of $0$. A mechanism that 
	always trades $0$ units at price $0$ is by definition equal to a mechanism 
	$\mathbb{M}_{0,\varnothing,\tau}$, where $\tau$ is arbitrary and irrelevant as 
	there is only a single outcome that the mechanism outputs. 
	
	Next, assume that $p > 0$. We prove the claim separately for $\mathcal{S}_k$ 
	and $\mathcal{I}_k$, and we start with $\mathcal{S}_k$.
	By Lemma \ref{lem:menustructure}, for every pair of valuations $(v,w)$ it holds 
	that $M_v = \{ q \in S : q \leq \min\arg_q'\max\{v(q') + p(q')\}\} \cup T$ where 
	$T$ is an arbitrary set of utility-maximizing quantities in $S$ for a buyer with 
	valuation $v$, and $N_w = \{q \in S : q \leq \min\arg_{q'}\max\{w(k-q') + 
	w(q)\}\} \cup T'$ where $T'$ is an arbitrary set of utility maximizing quantities 
	in $S$ for a seller with valuation $w$. Let $\tau_S(w)$ be the seller's utility 
	maximizing quantities in $N_w$ and let $\tau_B(v)$ be the buyer's 
	utility-maximizing quantities in $M_v$. If $\tau_S(w)$ and $\tau_B(v)$ intersect, 
	then by DSIC, mechanism must output any quantity in $\tau_S(w) \cap 
	\tau_B(v)$: call this quantity $\tau_{\cap}(v,w)$. Otherwise, if $\tau_S(w) \cap 
	\tau_B(v) = \varnothing$, the mechanism must output $\min\{\max \tau_S(w), 
	\max \tau_B(v)\}$, in order to satisfy the DSIC property: Assume that and 
	$\max \tau_B(v) > \max \tau_S(w)$ (the other case is symmetric) and suppose 
	that the mechanism trades any quantity $q \not= \max \tau_S(w)$. Since the 
	traded quantity $q$ must lie in the intersection of $M_v$ and $N_v$ and since 
	$\tau_B(v)$ and $\tau_S(w)$ are sets of highest quantities in $M_v$ and 
	$N_w$ respectively, we have $q < \max \tau_S(w)$. Hence, among the 
	quantities in $N_w$ a quantity less than $\max \tau_S(w)$ is traded, but the 
	buyer prefers quantity $\max \tau_S(w)$ because $\max \tau_S(w)$ is closer to 
	the buyer's set $\max \tau_B(v)$ of utility-maximizing quantities in $M_v$, 
	which would give the buyer a higher utility due to submodularity. The buyer 
	would thus misreport such that $\max \tau_S(w)$ is output instead. Note the 
	tie-breaking functions $\tau = (\tau_B,\tau_S,\tau_{\cap})$ we just established, 
	as well as the derived traded quantity $q$ given to the buyer, agree precisely 
	with those of Definition \ref{def:char}. We complete the equivalence by noting 
	that $0 \in S$ as we can define a seller's utility function that grows extremely 
	steeply up to quantity $k$, so that $(0,0)$ is the only IR outcome. This implies 
	that $\mathbb{M} = \mathbb{M}_{p,S\setminus\{0\},\tau}$.
	
	It remains to prove the claim for $\mathcal{I}$. Suppose for contradiction that 
	there are at least two positive quantities $q,q'$ in $S$, where $0 < q < q'$. 
	We apply the same technique as in Lemma \ref{lem:linearprice}. Let $(v,w)$ be 
	a valuation profile such that the mechanism selects $(q,p(q))$ when $(v,w)$ is 
	reported, and let $(v',w')$ be a valuation profile such that the mechanism 
	selects $(q',p(q'))$ when $(v',w')$ is reported. 
	
	Let $v^*$ be a valuation function that increases extremely slowly up to 
	quantity $q-1$, then jumps to a value of $pq + 2\epsilon$ at quantity $q$ and 
	proceeds again to grow extremely slowly up to quantity $q'-1$, and finally 
	jumps to a value of $pq' + \epsilon$ at quantity $q'$ after which it grows 
	extremely slowly onward. The only IR quantities that the mechanism can trade 
	when a buyer reports $v^*$ are $0$, $q$, and $q'$. As $(q,p(q)) \in N_w$, the 
	mechanism must select the outcome $(q,p(q))$ when $(v^*,w)$ is reported, 
	because of DSIC. So, $(q,p(q)) \in M_{v^*}$. Next, we construct a function 
	$w''$ for which it holds that $(q,p(q))\not\in N_{w''}$ and $(q',p(q')) \in 
	N_{w''}$: Valuation $w''$ is defined such that it increases extremely steeply 
	up to the quantity $k-q'$. Subsequently it increases by an amount of $p + 
	(q'+1)\epsilon$ to quantity $k-q'+1$, and it increases at a rate of $p - 
	\epsilon$ afterward. Note that the only IR quantities that can be traded under 
	$w''$ are $0$ and $q'$. We thus have that $(q,p(q))\not\in N_{w''}$ and 
	$(q',p(q')) \in N_{w''}$ because $(q',p(q'))$ in $M_{v'}$ and by DSIC the 
	mechanism must select $(q',p(q'))$ when $(v',w'')$ is reported. Therefore, 
	when $(v^*,w')$ is reported, $(q',p(q'))$ is selected so we see that $(q',p(q')) 
	\in M_{v^*}$.
	
	Let $w^*$ be a valuation function defined as follows. Let $\epsilon > 0$ be 
	sufficiently small. We let $w^*(k) = kp$, and for all $q'' > 0$ not equal to $q$ 
	or $q'$, We let $w^*(k-q'') = p\cdot (k-q'') - \epsilon$, so that the seller's 
	increase in utility for trading $q''$ units is $pq'' - (w^*(k) - w^*(k-q'')) = pq'' - 
	(pk - p(k-q'') + \epsilon) = -\epsilon$, so when $w^*$ is the valuation of the 
	seller, the mechanism cannot trade $q''$ items as that would violate IR. 
	Moreover, we define $w^*(k-q) = p\cdot(k-q)+\epsilon$ and $w^*(k-q') = 
	p\cdot(k-q') + 2\epsilon$, so that trading $q$ or $q'$ units leads to an 
	increase in utility for a seller with valuation $w^*$ and so that trading $q''$ 
	units is the preferred quantity to trade for a seller with valuation $w^*$. We 
	now see that $(q',p(q')) \in N_{w^*}$, because $(q',p(q'))$ is in $M_{v'}$ so 
	that by DSIC the mechanism outputs $(q',p(q'))$ on report $(v',w^*)$. Next, 
	we construct a function $v''$ for which it holds that $(q',p(q'))\not\in M_{v''}$ 
	and $(q,p(q)) \in M_{v''}$. This function is defined as follows: $v(q'') = 
	pq''-\epsilon$ for all $q''$ except $q$, where $v(q) = pq + \epsilon$. When a 
	buyer reports $v''$, by IR the mechanism can either trade $0$ or $q$ units 
	and no other quantity. On report $(v'',w)$ the mechanism must output 
	$(q,p(q))$ due to DSIC and because $(q,p(q)) \in N_w$ by assumption. Thus 
	$(q,p(q)) \in N_{v''}$ and $(q',p(q')) \not\in M_{v''}$, hence when $(w^*,v')$ is 
	reported the mechanism outputs $(q,p(q))$ because of DSIC. This establishes 
	$(p,p(q)) \in N_{w^*}$.
	
	We thus have constructed two functions $v^*$ and $w^*$ for which it holds 
	that both $(q,p(q))$ and $(q',p(q'))$ are in both $M_{v^*}$ and $N_{w^*}$. 
	Moreover, a buyer with valuation $v^*$ strictly prefers $(q,p(q))$ over 
	$(q',p(q'))$, so by DSIC the mechanism must output the outcome $(q,p(q))$ 
	when $(v^*,w^*)$ is reported. However, a seller with valuation $w^*$ strictly 
	prefers $(q',p(q'))$ over $(q,p(q))$, so by DSIC the mechanism must output the 
	outcome $(q',p(q'))$ when $(v^*,w^*)$ is reported, which is a contradiction. 
	So, we must refute the assumption that there are at least 2 quantities that the 
	mechanism can trade.
	
	Hence, either $(0,0)$ is always output, in which case the claim is trivial (the 
	mechanism is equal to $\mathbb{M}_{0,\varnothing,\tau}$, where $\tau$ is not 
	relevant), or there is a unique positive quantity $q$ such that the mechanism 
	selects either $(q,p(q))$ or $(0,0)$ and outputs $(q,p(q))$ on at least one 
	valuation profile $(v,w)$. It now suffices to prove, by definition of the 
	mechanism (Definition \ref{def:char}), that outcome $(q,p(q))$ is selected if 
	both players experience an increase in utility from this outcome. Let $(v',w')$ 
	be an arbitrary valuation profile for which the latter holds. As $(q,p(q)) \in 
	M_v$, we infer that $(q,p(q))$ is output on report $(v,w')$ so that $(p,p(q)) \in 
	M_w'$. Thus, by DSIC, the mechanism must select $(q,p(q))$ on report 
	$(v',w')$ as otherwise a buyer with valuation $v'$ would report $v$ instead. 
	This proves that the mechanism equals $\mathbb{M}_{ p,\{(q,p(q))\},\tau }$.
\end{proof}

\end{document}